\documentclass[10pt,journal,compsoc]{IEEEtran}

\usepackage{booktabs}
\usepackage{xcolor}
\usepackage{color}
\usepackage{bbding}
\usepackage{array, caption, threeparttable}

\usepackage{multirow}
\usepackage{graphicx}
\usepackage{subfigure}
\usepackage{amsmath}
\usepackage{amssymb}
\usepackage{amsthm}
\usepackage{url}
\usepackage{boxedminipage}
\usepackage{algorithmic}
\usepackage[linesnumbered,ruled,vlined]{algorithm2e}
\usepackage{multicol}
\usepackage{lettrine}
\usepackage{epstopdf}
\newcommand{\paratitle}[1]{\vspace{0.5ex}\noindent {\textbf{#1}}}

\newtheorem{theorem}{\bf Theorem}
\graphicspath{{figs/}}
\newcommand{\ie}{\emph{i.e.,}\xspace}
\newcommand{\eg}{\emph{e.g.,}\xspace}
\newcommand{\etc}{\emph{etc.}\xspace}

\newcommand{\eat}[1]{}

\begin{document}

\title{Pocket Diagnosis: Secure Federated Learning against Poisoning Attack in the Cloud}

\author{Zhuoran~Ma, Jianfeng~Ma, Yinbin~Miao,
        Ximeng~Liu,~\IEEEmembership{Member,~IEEE}, Kim-Kwang Raymond~Choo,~\IEEEmembership{Senior Member,~IEEE,} and Robert H. Deng,~\IEEEmembership{Fellow,~IEEE}
\IEEEcompsocitemizethanks{
\IEEEcompsocthanksitem Z. Ma, J, Ma and Y. Miao (corresponding author) are with the School of Cyber Engineering, Xidian University, Xi'an 710071, China; Guangxi Key Laboratory of Cryptography and Information Security, Guilin 541004, China; Shaanxi Key Laboratory of Network and System Security, Xidian University, Xi'an 710071, China. E-mail: emmazhr@163.com, jfma@mail.xidian.edu.cn, ybmiao@xidian.edu.cn
\IEEEcompsocthanksitem X. Liu is with the Key Laboratory of Information Security of Network Systems, College of Mathematics and Computer Science, Fuzhou University, Fuzhou 350108, China. Email: snbnix@gmail.com
\IEEEcompsocthanksitem K.-K. R. Choo is with the Department of Information Systems and Cyber Security, The University of Texas at San Antonio, San Antonio, TX 78249 USA. Email: raymond.choo@fulbrightmail.org
\IEEEcompsocthanksitem R. H. Deng is with the School of Information Systems, Singapore Management University, Singapore 188065. Email: robertdeng@smu.edu.sg}
}


\IEEEtitleabstractindextext{
\begin{abstract}
    Federated learning has become prevalent in medical diagnosis due to its effectiveness in training a federated model among multiple health institutions (\ie Data Islands (DIs)). However, increasingly massive DI-level poisoning attacks have shed light on a vulnerability in federated learning, which inject poisoned data into certain DIs to corrupt the availability of the federated model. Previous works on federated learning have been inadequate in ensuring the privacy of DIs and the availability of the final federated model. In this paper, we design a secure federated learning mechanism with multiple keys to prevent DI-level poisoning attacks for medical diagnosis, called SFPA. Concretely, SFPA provides privacy-preserving random forest-based federated learning by using the multi-key secure computation, which guarantees the confidentiality of DI-related information. Meanwhile, a secure defense strategy over encrypted locally-submitted models is proposed to defense DI-level poisoning attacks. Finally, our formal security analysis and empirical tests on a public cloud platform demonstrate the security and efficiency of SFPA as well as its capability of resisting DI-level poisoning attacks.
\end{abstract}

}

\maketitle

\IEEEraisesectionheading{\section{Introduction}}

\lettrine[lines=2]{W}{ith} the exponential growth of demands for medical diagnosis, machine learning has been ubiquitously used to drive critical decisions in telemedicine diagnosis~\cite{kononenko2001machine,ZXLYZL18,wang2019privacy}. One depressing situation in training a diagnosis model arises as relatively limited medical data are generated and collected in a single health organization~\cite{friedman2010achieving}. Most of medical data (\ie non-independent and identically distributed (non-i.i.d.) data~\cite{homem2008rates}) are distributed among multiple institutions in the form of ``Data Island" (DI). The training over local data will incur the underfitting in the predictive model as the training data is insufficient, thereby leading to incorrect medical diagnosis.

To embrace the advantages (\eg cost savings, access feasibility) of cloud computing, patients prefer to remotely transfer medical data for diagnosis at any time, as convenient as ``put a doctor into pockets", which is figuratively described as pocket diagnosis.
To improve a more accurate diagnosis model, all patients can outsource their individually-collected medical data to a cloud server, but will compromise the data privacy as the medical data generally contain much sensitive information such as age, sex and family heredity history.
It is urgently required for privacy protection~\cite{miao2019fair,miao2019optimized,miao2019secure,miao2019privacy}. With the ever-increasing awareness of data privacy~\cite{liu2018efficient,miao2019multi,miao2018enabling}, European Union enforced the GDPR~\cite{gdpr} to safeguard data security without roughly gathering all data from multiple DIs for traditional machine learning.

Motivated by privacy concerns of distributed learning in the cloud, the federated learning~\cite{Liu2018Secure,Dai2018Privacy,Segev2016Learn,brisimi2018federated} has been proposed as an efficient method to train a federated model by aggregating outsourced local models trained on multiple DIs in a privacy-preserving way. Since local models contain sensitive information of DIs, each DI encrypts individual local model under his/her key before outsourcing it to the cloud~\cite{aono2017privacy}.
Existing federated learning schemes are generally deployed in the single key setting~\cite{mcmahan2016communication,wang2019beyond,Liu2018Secure}, in which all DIs use the same key pair to respectively encrypt their models. However, if one DI is comprised by an adversary, it will lead to the privacy leakage in the pocket diagnosis.
The key agreement~\cite{xu2019verifynet} avoids this kind of privacy leakage by creating a shared key between any two DIs, but will incur high communication overheads as the number of DIs increases.
\textit{The first challenge is how to avoid the security threat and unnecessary overhead in pocket diagnosis under the single-key scenario.}

Thus, it is necessary to consider a more suitable model for secure computation to achieve the tradeoff between the availability and security of federated learning. Meanwhile, to provide accurate diagnosis services, an interpretable diagnosis model is required that can produce insights about the causes of decisions~\cite{gilpin2018explaining,nishio2018computer}.

It is also worth mentioning the corruption of training data in pocket diagnosis. Recently, the attack of injecting carefully crafted adversarial samples into the training dataset, so-called poisoning attack~\cite{Suciu2018MLF,biggio2018wild,al2019privacy}, has been proposed. The attacker aims to lead classification errors during the diagnosis time. Generally speaking, a wrong medical diagnosis may lead to a loss of trust of users and life-threatening consequences~\cite{mozaffari2015systematic}. The federated learning can protect the privacy of local models,
but it cannot guarantee that none of DIs are malicious~\cite{bagdasaryan2018backdoor,papernot2018sok}. Thus, the way may lead to DI-level poisoning attacks~\cite{bhagoji2018analyzing}.
Fig.~\ref{introdistribution} gives a high-level overview of the DI-level poisoning attack in the pocket diagnosis, the attacker injects poisoned data into certain DI (\ie ${DI}_2$) to poison the federated model. If the attack succeeds without timely defense, the poisoned federated model will return a wrong diagnosis result for the submitted request. To eliminate the above-mentioned hazards, a range of defense countermeasures have been designed to prevent poisoning attacks.
However, prior works~\cite{klivans2009Learning,steinhardt2017certified,jagielski2018manipulating} on poisoning attacks focus on implementing over local data stored in a single DI, but these schemes are futile for DI-level poisoning attacks as intermediate parameters will reveal local data privacy.
\textit{The second challenge is how to protect the training accuracy of pocket diagnosis under the adversarial setting of poisoning attacks.}

Due to the shortcomings (\ie limitations of local data access and the lack of privacy protection) of previous defense strategies, existing defense methods are not suitable for federated learning. To guarantee privacy and availability of the federated learning, it is urgent to propose a secure defense countermeasure against DI-level poisoning attacks in the federated learning.

\subsection{Our Contributions}
To address above challenges, in this paper we present a \textbf{S}ecure \textbf{F}ederated learning against DI-level \textbf{P}oisoning \textbf{A}ttack for the pocket diagnosis under multiple keys, which is termed as SFPA. Our main contributions are summarized as follows.

\begin{figure}
  \centering
  \includegraphics[width=3.5in]{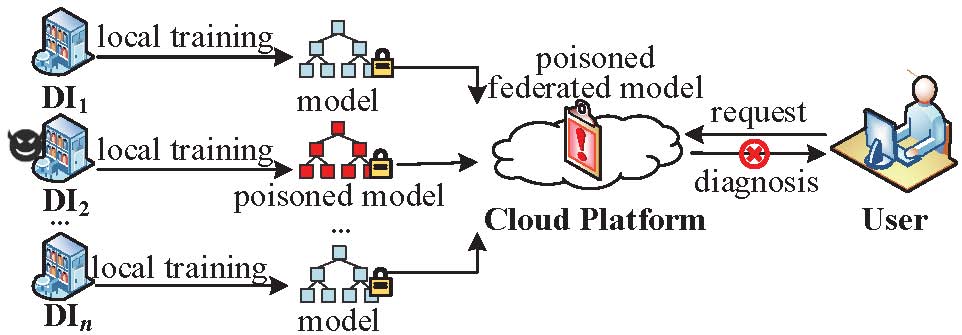}
  \caption{Overview of DI-level poisoning attack.}
  \label{introdistribution}
\end{figure}

\begin{itemize}
  \item \textit{Secure Federated Learning with Multiple Keys}: To prevent the security threat caused by the leakage of a single key, SFPA proposes a privacy-preserving federated learning framework based on multi-key computation. The framework designs a random forest-based federated diagnosis model over multiple DIs, which can be executed on a single cloud server to avoid the collusion of multi-server model.

  \item \textit{Secure Defense Countermeasure}: To defense DI-level poisoning attacks, SFPA proposes a secure and effective defense countermeasure, which ensures the privacy of DIs in the multi-key setting. The defense approach based on trimmed optimization can resist a range of DI-level poisoning attacks in the federated learning.
  \item \textit{Privacy-preserving Pocket Diagnosis}: To provide diagnosis services, SFPA designs th multi-key pocket diagnosis, which allows users to remotely send medical requests via mobile devices and instantly obtain final encrypted diagnosis results. 
\end{itemize}

The remainder of this paper is organized as follows. Section~\ref{sec:rekatedwork} and Section~\ref{sec:preliminaries} introduce related works and relevant technology. In Section~\ref{sec:system model}, we give a formal description of the system model, threat model and design goals. Then, we elaborate the construction of the {SFPA} framework in Section~\ref{secframework}, followed by the security analysis of SFPA system in Section~\ref{sec:privacyAnal}. Section~\ref{sec:perf} presents performance evaluation of SFPA sysetm. Finally, we conclude this paper and discuss its future work in Section~\ref{sec:conclusion}.

\section{Related Work}\label{sec:rekatedwork}
Previous studies on federated learning are based on the confidentiality of locally-trained models,
but do not take the availability of a federated model in the adversarial world into consideration.
In order to solve existing problems, we outline a brief review of previous studies on privacy-preserving federated learning and defense countermeasures against poisoning attacks, respectively.
\subsection{Secure Computation with Federated Learning}
The primitive of federated learning was first pointed out by Google~\cite{konevcny2016federated}, which is used to train the federated model over locally-trained models. Compared with the conventional training methods, the federated learning provides the privacy preservation by leveraging homomorphic encryption, and achieves the high efficiency with implementing local training over plaintexts. After that, many federated learning schemes were proposed to improve the efficiency and security~\cite{mcmahan2016communication,wang2019beyond}. McMahan \textit{et al.}~\cite{mcmahan2016communication} designed a secure aggregation algorithm based on Paillier cryptosystem on deep networks. Wang \textit{et al.}~\cite{wang2019beyond} proposed a federated learning framework constructed by using the convolutional neural network and a Paillier cryptosystem. However, these schemes only support the single key setting. Once the encrypted data are leaked or the communication channel is compromised, it is easy to obtain the secret key to decrypt the private data. To solve the dilemma, existing works on traditional machine learning~\cite{liu2016efficient,ma2019privacy,liu2018efficient} were proposed for privacy protections in the multi-key settings. Liu \textit{et al.}~\cite{liu2016efficient,liu2018efficient} designed multi-key secure computation protocols. Ma \textit{et al.}~\cite{ma2019privacy} presented a privacy-preserving Random Forest (RF) scheme for distributed learning, where RF with the natures of interpretability and transparency has been widely used in medical diagnosis. All the above schemes~\cite{liu2016efficient,ma2019privacy,liu2018efficient} can protect data privacy with multiple non-colluding cloud servers. Even if one cloud server is corrupted, it still can guarantee data privacy. Unfortunately, it is difficult to employ multiple non-colluding servers in the real-world due to commercial competitions and conflicts of interest.

There exists security problems in the single-key setting and constraints of a strong assumption in the multi-key setting.
Therefore, it is necessary to propose a multi-key federated learning scheme in the multiple DIs, which aims to eliminate the collusion assumption of multi-server model and guarantee multi-key secure computation.

\subsection{Defense Countermeasures against Poisoning Attack}
Practical poisoning attacks have been demonstrated in the machine learning. To defense such attacks, an increasing number of defense strategies have been presented. Klivans \textit{et al.}~\cite{klivans2009Learning} first introduced a learning framework, which removes the malicious data to avoid outliers by using a clustering algorithm. The similar defense strategies were suggested~\cite{papernot2018sok,steinhardt2017certified}, which prevent poisoned data in a training dataset by minimizing the loss function in the model.
If the error rate of a training sample is higher than a threshold, then the sample will be considered as an outlier. However, as the real distribution of benign training data is obviously unknown, it is extremely difficult to identify all poisoned data. Instead of removing outliers from the training data, a trimmed optimization is deployed to make machine learning robust. Jagielski \textit{et al.}~\cite{jagielski2018manipulating} adopted a trimmed loss function to select a subset of training data with the lowest residuals relative to the training model, removing data with large residuals to trim training data. The subset is close to the benign training data which does not contribute much to the poisoning attack. Wang \textit{et al.}~\cite{wang2019neural} proposed a novel adversarial deep neural network model by selecting a data subset to detect and mitigate poisoning attacks. Aforementioned strategies can effectively defense poisoning attacks in the regression learning and deep learning, but are restricted with data poisoning in a single DI. Besides, computational intermediate parameters of the defense process may lead to privacy leakage. Unlike data poisoning on traditional machine learning, the federated learning focuses on the issue of model poisoning~\cite{bhagoji2018analyzing}, resisting poisonous local-models from certain malicious DIs. Considering both security and availability of pocket diagnosis, it is urgent to propose a secure federated learning scheme for pocket diagnosis, which supports the multi-key setting and resists DI-level poisoning attacks.

\begin{table}[!ht]
 \centering
 \scriptsize
 \caption{Functionalities, securities and techniques in various schemes: A comparative summary}
 \tabcolsep 6pt
 \begin{tabular*}{3.5in}{l|cccccc}
  \toprule \hline
  Schemes &\cite{wang2019beyond} &\cite{xu2019verifynet} &\cite{ma2019privacy} &\cite{jagielski2018manipulating}&\cite{wang2019neural} &SFPA \\
  \hline
  Federated Learning &\Checkmark &\Checkmark &\XSolidBrush &\XSolidBrush&\XSolidBrush &\Checkmark \\\hline
  Secure Computation &\Checkmark &\Checkmark &\Checkmark &\XSolidBrush &\XSolidBrush&\Checkmark \\\hline
  Server&Single &Multiple &Dual &--- &Single &Single \\\hline
  Multi-key Setting &\XSolidBrush &\XSolidBrush &\Checkmark &\XSolidBrush&\XSolidBrush &\Checkmark \\\hline
  Defense Strategy  &\XSolidBrush &\XSolidBrush &\XSolidBrush &\Checkmark&\Checkmark &\Checkmark  \\\hline
  \bottomrule
  \end{tabular*}
  \label{campare}

\begin{tablenotes}
\item \textbf{Notes}. ``Single" and ``Dual" represent the number of cloud servers, respectively; ``---" means that is not involved secure computation.
\end{tablenotes}
\end{table}

TABLE~\ref{campare} summarizes the comparison between our SFPA and previous schemes~\cite{wang2019beyond,liu2016efficient,ma2019privacy,jagielski2018manipulating} in terms of several aspects (i.e., secure computation, multi-key setting, \etc). It reveals that SFPA provides both secure federated learning in the multi-key setting and defense against DI-level poisoning attacks. Meanwhile, our SFPA avoids the collusion of the multi-server model.

\section{Perliminary}\label{sec:preliminaries}

This section describes the definition of random forest, two-trapdoor public-key cryptosystem for multi-key decryption~\cite{liu2018efficient} and multi-key secure computation~\cite{liu2016efficient}.

\subsection{Random Forest}

Random Forest (RF) is an ensemble learner with outstanding interpretability, which consists of a series of Decision Trees (DTs). Assume that a RF includes $t$ DTs (\ie $rf=\{dt_1,dt_2,...,dt_t\}$), a label class is defined as $(c_0=-1,c_1=1)$\footnote{Without loss of generality, here we only consider binary classification.}. The final RF prediction result of a sample $x$ is computed as $\hat{y}=\frac{1}{t}\sum_{i=1}^{t}dt_i(x).$
When the value $\hat{y}\geq0$, the prediction is $c_1$; otherwise, the prediction is $c_0$. In a DT, the weight $w$ of a splitting node is the splitting value, the weight $w$ of a leaf node is the class label.

\subsection{Multi-key Decryption Scheme}\label{PublicKeyCrypto}
Here, we introduce the multi-key decryption scheme, which consists of five algorithms.
\begin{itemize}
  \item \texttt{KeyGen}: Given the security parameter $\varsigma$, the public key $pk=(N,g)$ and secret key $sk=\lambda$ are computed, where the secret key $sk$ can be randomly split into secret shares $sk^{(1)}$ and $sk^{(2)}$.
\item \texttt{Enc$_{pk}(m)$}: Given a plaintext $m$, the encrypted data $[\![m]\!]$ is computed with the public key $pk$.
\item \texttt{Dec$_{sk}([\![m]\!])$}: Given the ciphertext $[\![m]\!]$, the plaintext $m$ is decrypted with the secret key $sk$.
\item \texttt{SDec$_{sk^{(i)}}([\![m]\!])$}: Given the ciphertext $[\![m]\!]$, the decryption share $[\![m]\!]^{(i)}=[\![m]\!]^{{sk}^{(i)}}\mod N^2$ is computed with the corresponding secret share $sk_{i}$.
\item \texttt{WDec$(\{[\![m]\!]^{(1)},[\![m]\!]^{(2)}\})$}: Given the tuple of decryption shares $\{[\![m]\!]^{(1)},[\![m]\!]^{(2)}\}$, the plaintext $m$ is decrypted.
\end{itemize}

\subsection{Multi-Key Secure Computation}\label{sec_computation}

Here, we introduce the Secure Multiplication (\texttt{SMUL}), Secure Comparison (\texttt{SCOM}), Secure Addition (\texttt{SADD}), Secure Transformation (\texttt{STRA}) operations of multi-key secure computation. Given encrypted data under two different keys (\ie ${pk_a}, {pk_b}$), the computation result under the other public key $pk$ can be computed as follows:
\begin{itemize}
  \item \texttt{\texttt{SMUL$([\![m_a]\!]_{pk_a},[\![m_b]\!]_{pk_b})$}}: It outputs the multiplication result $[\![m_a \times m_b]\!]_{pk}$.
  \item \texttt{\texttt{SCOM$([\![m_a]\!]_{pk_a},[\![m_b]\!]_{pk_b})$}}: It outputs the comparison result $res$ to show which number is bigger. If $res=0$, $m_a \geq m_b$; Otherwise, $m_a < m_b$.
  \item \texttt{\texttt{SADD$([\![m_a]\!]_{pk_a},[\![m_b]\!]_{pk_b})$}}: It outputs the sum $[\![m_a + m_b]\!]_{pk}$.

  \item \texttt{\texttt{STRA$([\![m]\!]_{pk_a})$}}: It transforms $[\![m]\!]_{pk_a}$ to $[\![m]\!]_{pk}$.
\end{itemize}

\section{Problem Formulation}\label{sec:system model}

In this section, we define the system model, threat model and design goals of SFPA, respectively.
\subsection{System Model}
\begin{figure}
  \centering
  \includegraphics[width=3.5in]{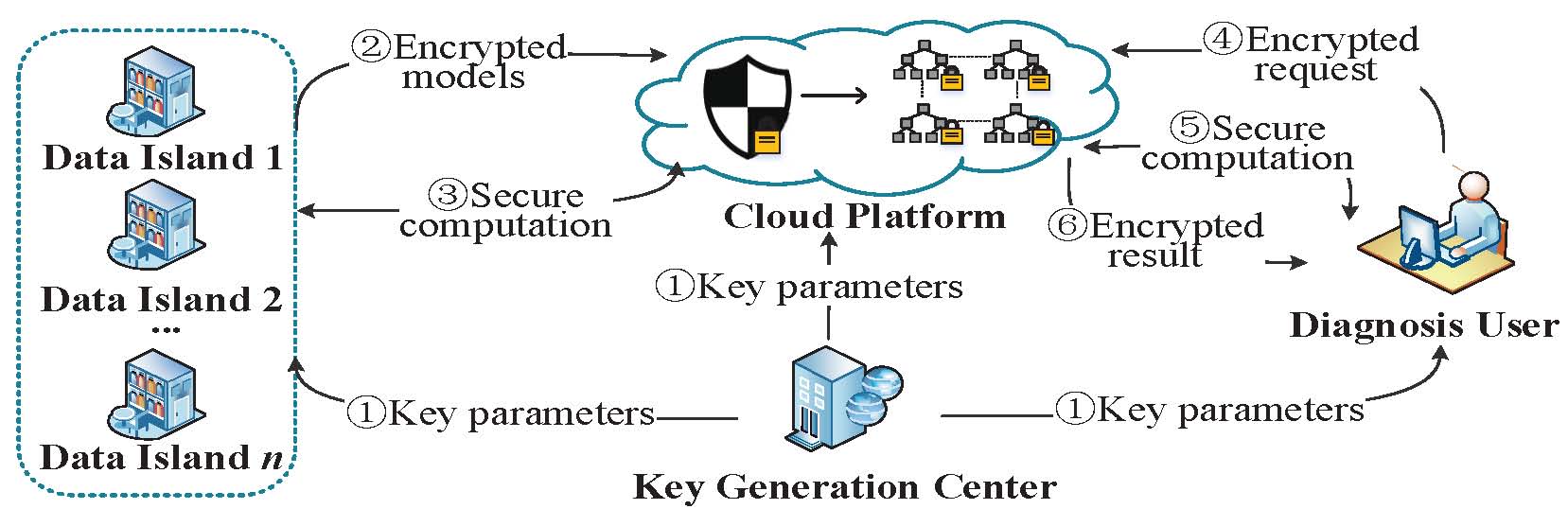}
  \caption{System model}
  \label{system_model}
\end{figure}
As demonstrated in Fig.~\ref{system_model}, our system model consists of four entities: Key Generation Center (KGC), Cloud Platform (CP), Data Island (DI) and Diagnosis User (DU). The concrete role of each entity is shown as follows:
\begin{itemize}
  \item \textit{Key generation center}. KGC is in charge of generation, distribution, and management of all keys in our system.
  \item \textit{Data island}. Each DI hosting distributed medical data first trains local models on individual data, and then encrypts local models by using the individual key before sending them to CP. Without loss of generality, the local training data of all DIs are considered as non-i.i.d. distribution. Besides, each DI conducts multi-key secure computation operations by cooperating with CP.
  \item \textit{Cloud platform}. CP owns unlimited storage space for encrypted locally-trained models, and implements the secure defense against DI-level poisoning attacks. Then, benign local models are securely aggregated to generate a federated model. Besides, CP can conduct multi-key secure computation operations.
  \item \textit{Diagnosis user}. DU first encrypts a diagnosis request with his/her public key before submitting it to CP, then decrypts the diagnosis result returned by CP. Besides, DUs are involved with multi-key secure computations that are executed with CP.
\end{itemize}

\subsection{Threat Model}~\label{threatModel}
In our threat model, we assume that KGC is trustable for key management. DU, DI, and CP are considered as \textit{honset-but curious} entities, which honestly follow pre-defined protocols but are interested in deducing private information from other entities. Besides, we assume that DI, CP, and DU are not in collusion. According to the available information and capabilities to an adversary, we consider two attack models with different attack abilities:
\begin{itemize}
  \item \textit{Privacy violation}. The adversary can observe encrypted locally-trained models, intermediate calculation parameters, encrypted diagnosis requests and returned results.

  \item \textit{Poisoning attack}. Suppose the adversary can access local training dataset $D_{tr_i}$ hosted on certain DIs and inject poisonous data to corrupt the locally-trained models, eventually achieve the DI-level poisoning attack on federated learning.
\end{itemize}

\subsection{Design Goals}\label{designGoal}

To achieve privacy-preserving federated learning and resist DI-level poisoning attacks under the aforementioned adversarial environment, our design must meet security and accuracy requirements
as follows.
\begin{itemize}
  \item \textit{Security}. SFPA should protect the confidentiality of locally-trained models and the privacy of intermediate computational parameters on federated learning. Besides, the proposed defense countermeasure against DI-level poisoning attacks should not reveal any private information. Finally, in the diagnosis phase, both DU's submitted-requests and corresponding diagnosis results should be only known to DU.
  \item \textit{Accuracy}. SFPA must guarantee the diagnosis accuracy of a federated model. Under the adversarial environment, SFPA framework can resist DI-level poisoning attacks, where a federated model is able to make correct predictions for pocket diagnosis.
\end{itemize}

\section{SFPA Framework with Multiple Keys}\label{secframework}

\begin{figure}[!ht]
  \centering
  \includegraphics[width=3.1in]{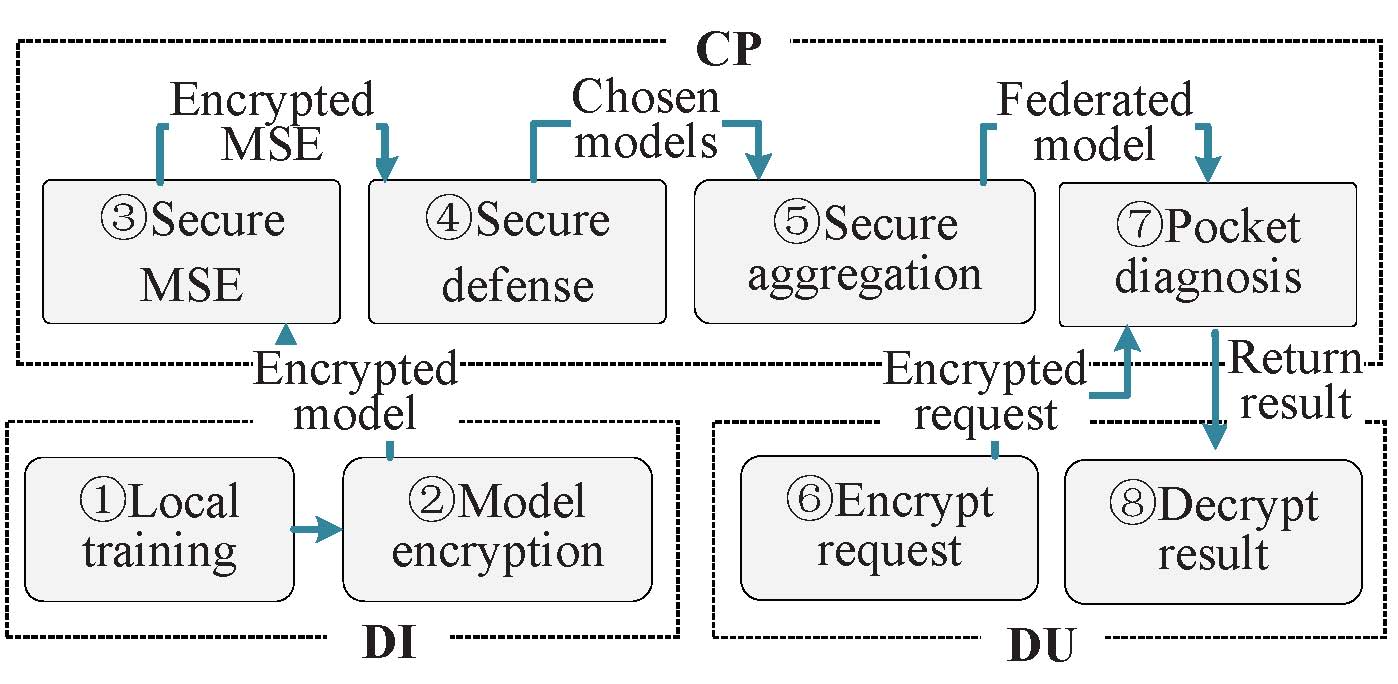}
  \caption{Overview of SFPA}
  \label{overview}
\end{figure}

Traditional federated learning schemes incur high computation and storage burden on constructing a diagnosis model, and even cannot resist poisoning attacks.
In this section, we first present a privacy-preserving RF-based federated learning framework in the ideal world to provide an efficient and secure training method, then design a secure defense countermeasure to resist DI-level poisoning attacks in the adversarial world. Finally, we securely implement pocket diagnosis.
The overview of SFPA is shown in Fig.~\ref{overview}, there are four main processes: initialization (Step $\textcircled{1}$$\textcircled{2}$), secure defense countermeasure (Step $\textcircled{3}$$\textcircled{4}$), privacy-preserving RF-based federated learning (Step $\textcircled{5}$), and privacy-preserving pocket diagnosis (Step $\textcircled{6}$-$\textcircled{8}$). The specified processes are defined as follows. Besides, some notations used in SFPA framework are shown in TABLE~\ref{notation}.

\subsection{Initialization}\label{sec_ini}

The initialization contains three phases: key distribution, model building, and model encryption. The concrete process is defined as follows.

\begin{table}[!ht]
 \centering
 \scriptsize
 \caption{Notation descriptions}
 \tabcolsep 3pt
 \begin{tabular*}{3.5in}{l|l}
 \toprule
 \hline
  Notations & Descriptions \\
 \hline
 ${\times}$ & Poisonous symbol\\ \hline
  ${*}$ & Benign symbol\\\hline
  $\textsf{FM}$ & Federated model \\\hline
  $\{sk^{(1)},sk^{(2)}\} $ & Secret shares for CP\\\hline
  $\{{sk}_{\text{DI}}^{(1)},sk_{\text{DI}}^{(2)} \} $ & Secret shares for a DI\\\hline
  $\{sk_u^{(1)},sk_u^{(2)} \}, sk_u$ & Secret shares and secret key for DU\\\hline
  ${[\![x]\!]_{pk}^{(1)}},{[\![x]\!]_{pk}^{(2)}}$ & Decryption shares for $[\![x]\!]_{pk}$\\\hline
  $F=\{f_1,f_2,...,f_{\kappa}\}$ & Feature set with the size of $\kappa$ \\\hline
  ${[\![{D}_{val}]\!]_{pk}}$ & Encrypted validation dataset  \\\hline
  $D_{tr_i}$ & Local training data belongs to $i$-th DI\\\hline
  $\textsf{Local}_{ben}^{\ast}$ & Choose locally-encrypted RF models to construct $\textsf{FM}$\\\hline
  $n^{*}$ & Number of benign data islands \\$n$ & Total number of data islands \\\hline
  $[\![rf]\!]_{pk_{\text{DI}_i}}$ & Locally-encrypted model belongs to $i$-th DI\\\hline
 $\textsf{Local}=\{[\![rf]\!]_{pk_{\text{DI}_i}}\}_{i=1}^n$ & Locally-encrypted models\\
\hline
 \bottomrule
 \end{tabular*}
 \label{notation}

\end{table}

\paratitle{KeyManagement}: Assume that SFPA framework contains $n$ DIs. KGC first generates keys $(pk,sk,sk^{(1)}, sk^{(2)})$, $(pk_{\text{DI}_i},sk_{\text{DI}_i},sk_{\text{DI}_i}^{(1)}, sk_{\text{DI}_i}^{(2)})$ and $(pk_u,sk_u,sk_u^{(1)}, sk_u^{(2)})$, where $sk^{(1)}, sk^{(2)}$ are corresponding secret shares of $sk$ and $sk_{\text{DI}_i}^{(1)}, sk_{\text{DI}_i}^{(2)}$ are corresponding secret shares of $sk_{\text{DI}_i}$ $(i=1,2,...,n)$, and $sk_u^{(1)}, sk_u^{(2)}$ are corresponding secret shares of $sk_u$, respectively. Besides, we define $pk=(N,g)$, $pk_{\text{DI}_i}=(N,g_i)$ and $pk_u=(N,g_u)$. Then, KC distributes keys to CP, DIs and DU, respectively. Concretely, details are shown in Fig.~\ref{figkeydistribute}.

\paratitle{ModelBuild}: As shown in Step $\textcircled{1}$, the $i$-th DI locally builds the RF model $rf_i$ over a local training dataset $D_{tr_i}$, where the trained model $rf_i$ consists of $t_i$ trees,  $rf_i=\{dt_1,dt_2,...,dt_{t_i}\}$.

\paratitle{ModelEnc}: As shown in Step $\textcircled{2}$, the locally-trained model $rf_i$ will be encrypted by the public key $pk_{\text{DI}_i}$ of the $i$-th DI, and each tree $dt\in rf_i $ can be encrypted from the root until the leaf node as
\begin{equation}
\begin{aligned}\label{modelEnc}
\nonumber
[\![dt]\!]_{pk_{\text{DI}_i}}=\{[\![w_{root}]\!]_{pk_{\text{DI}_i}},[\![w_1]\!]_{pk_{\text{DI}_i}},...,[\![w_{\Gamma}]\!]_{pk_{\text{DI}_i}}\},
\end{aligned}
\end{equation}
where $w$ (including split value and feature index) is a tree node's parameter, ${root}$ denotes the root node, and ${\Gamma}$ denotes the last node in a tree. Then, the $i$-th DI sends the encrypted model $[\![rf_i]\!]_{pk_{\text{DI}_i}}$ to CP.

\subsection{Privacy-preserving RF-based Federated Learning}\label{ppfederatedLearning}

In the ideal world, all DIs outsource encrypted locally-trained RF models under multiple keys, where DIs own mutually-different public keys ${pk_{\text{DI}_i}}$ to protect DIs' privacy (see Fig.~\ref{federatedlearning}(a)). We propose a privacy-preserving federated learning mechanism aiming at two targets: a) Constructing an accurate federated model over multiple RF models trained on DIs without accessing local data. b) Providing privacy protection for the whole process.

After the Step $\textcircled{2}$, locally-trained RF models are contained in $\textsf{Local}=\{[\![rf_i]\!]_{pk_{\text{DI}_i}}\}_{i=1}^{n}$ that are encrypted under different keys. Once receiving $\textsf{Local}$, CP will construct the federated model $\textsf{FM}$. Since $\textsf{Local}$ is encapsulated with mutually-different keys, it is necessary to aggregate $\textsf{Local}$ into $\textsf{FM}$ under the same key $pk$ for privacy preservations. The concrete process is shown as follows:

\paratitle{\textbf{SecureAggregation}}$(\rm\textsf{Local})\rightarrow \textsf{FM}$: As shown in Step $\textcircled{5}$, given encrypted models $\textsf{Local}$, CP transforms ciphertexts $\textsf{Local}$ under different keys $pk_{{\text{DI}_i}}(1\leq i\leq n)$ into ciphertexts $\textsf{FM}$ under the same public key $pk$ based on multi-key secure transformation (\texttt{STRA}). As the RF model from the $i$-th DI consists of $t_i$ trees, namely $[\![rf_i]\!]_{pk_{\text{DI}_i}}=\{[\![dt_1]\!]_{pk_{\text{DI}_i}},[\![dt_2]\!]_{pk_{\text{DI}_i}},...,[\![dt_{t_i}]\!]_{pk_{\text{DI}_i}}\}$, CP will use the \texttt{STRA} to transform each tree $[\![dt]\!]$ from the root node until leaf nodes as $[\![w]\!]_{pk}\leftarrow \texttt{STRA}([\![w]\!])$,\footnote{For simplification, we denote $[\![x]\!]_{pk_{\text{DI}_i}}$ as $[\![x]\!]$.} where the weight $[\![w]\!]$ of certain tree node will be transformed to $[\![w]\!]_{pk}$ under the sample public key $pk=(N,g)$. Since \texttt{STRA} is required to run over two parties, the $i$-th DI cooperates with CP to conduct secure aggregation on a tree node, which is demonstrated as follows:
\begin{itemize}
  \item  CP first chooses a random number $r\leftarrow\mathbb{Z}_N$ to blind $[\![w]\!]$ as $[\![w']\!]\leftarrow [\![w]\!]\cdot [\![r]\!]$, then computes the decryption share with $[\![w']\!]^{(1)}\leftarrow \texttt{SDec}_{sk_{\text{DI}}^{(1)}}([\![w']\!])$, finally sends both $[\![w']\!]$ and $[\![w']\!]^{(1)}$ to $i$-th DI.

  \item Then, the $i$-th DI first gets the decryption share via $ [\![w']\!]^{(2)}\leftarrow \texttt{SDec}_{sk_{\text{DI}}^{(2)}}([\![w']\!])$, then runs $w'\leftarrow \texttt{WDec}([\![w']\!]^{(1)},[\![w']\!]^{(2)})$ to obtain the decryption $w'$, finally returns $[\![w']\!]_{pk}\leftarrow \texttt{Enc}_{pk}(w')$ to CP.

  \item Finally, CP removes the blinded random number $r$ via $[\![w]\!]_{pk}\leftarrow [\![w']\!]_{pk} \cdot {[\![r]\!]_{pk}}^{N-1}$ to obtain the transformation result $[\![w]\!]_{pk}$.

\end{itemize}

\begin{figure}
  \centering
  \includegraphics[width=2.4in]{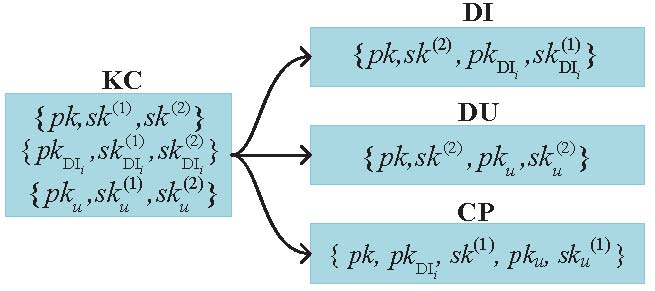}\\
  \caption{Key distribution in SFPA.}\label{figkeydistribute}
\end{figure}

In this way, each encrypted tree $[\![dt]\!]_{pk_{\text{DI}_i}}$ can be transformed into $[\![dt]\!]_{pk}$. Therefore, $[\![rf]\!]_{pk_{\text{DI}_i}}\in \textsf{Local}$ can be transformed into $[\![rf]\!]_{pk}$, and a federated model $\textsf{FM}$ is obtained via \textbf{SecureAggregation}$(\textsf{Local})\rightarrow \textsf{FM}$.

\subsection{Secure Defense Countermeasure}\label{sec_defense}

\begin{figure}[!ht]
	  \centering
      \subfigure[The view of ideal world.]{\includegraphics[width=1.73in]{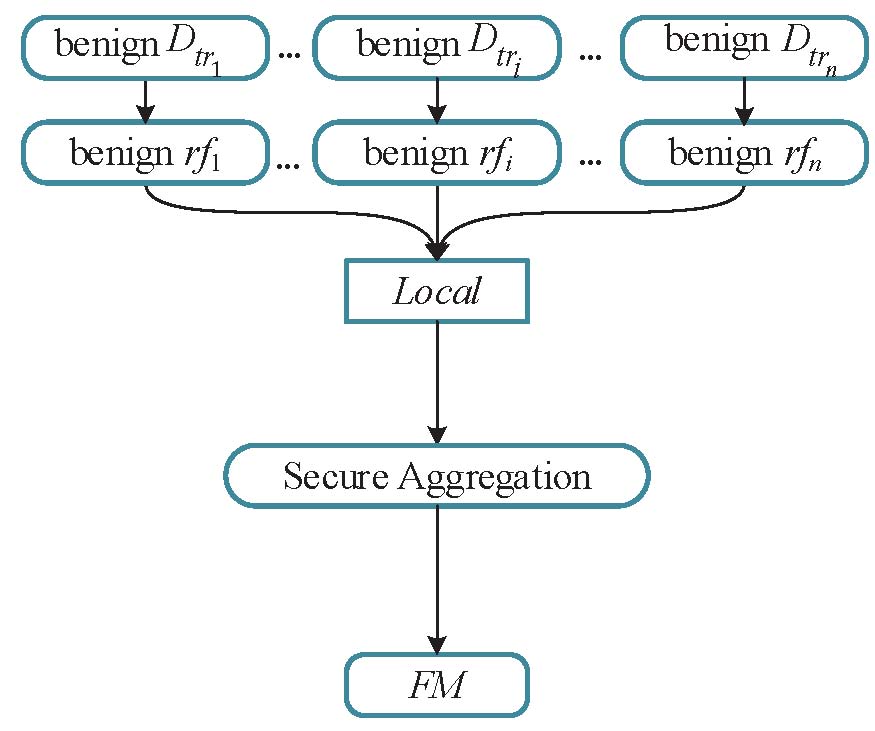}}
      \subfigure[The view of adversarial world.]{\includegraphics[width=1.73in]{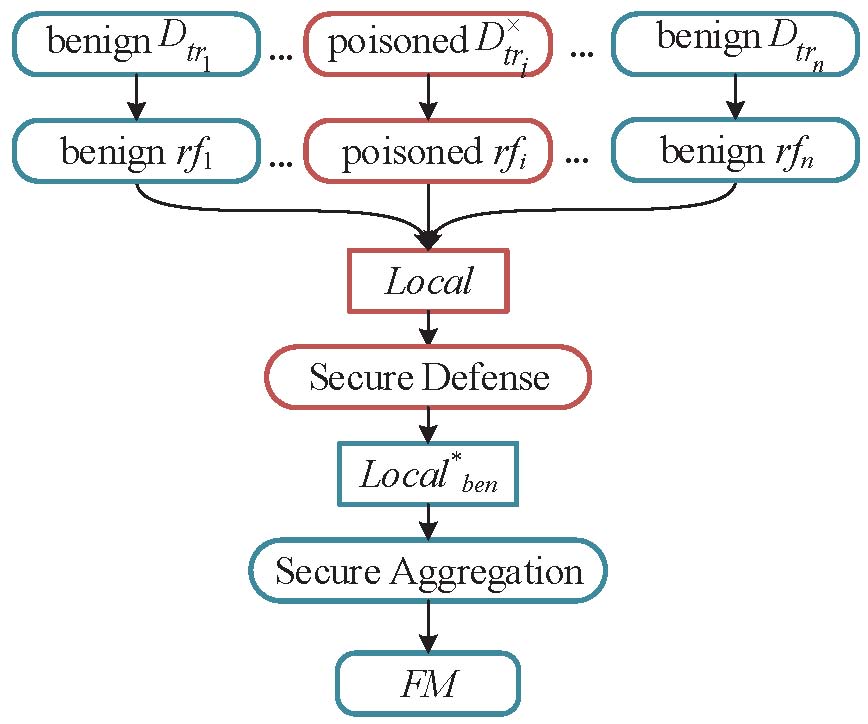}}
      \caption{Illustration of the comparison between ideal world and adversarial world in the privacy-preserving federated learning. There are $n$ DIs, ``$D_{tr_i}$" represents the training data belonging to the $i$-th DI, which is poisoned. } \label{federatedlearning}
\end{figure}

Different from the ideal world, the attacker in the adversarial world can inject poisonous data into DIs to launch a DI-level poisoning attack for the corruption of a federated model. Fig.~\ref{federatedlearning}(b) shows the view of adversarial world. Here, we propose a secure defense countermeasure to resist the DI-level poisoning attack without privacy leakage.

In the adversarial environment, we treat the attack as white-box, the attacker can know the training data ${D}_{tr}$, feature set ${F}$, learning algorithm ${H}$ and even training parameters $w$ belonging to $DI^{\times}$, the attacker's knowledge is described by
\begin{equation}
\begin{aligned}\label{RF_prediction}
knowledge=\{{D}_{tr},{F},{H},w\}.
\end{aligned}
\end{equation}

In addition, we define the DI-level poisoning attack as a bilevel optimization problem under the white-box knowledge assumption, as demonstrated in
\begin{equation}
\begin{aligned}\label{attackstrategy}
&argmax_{{D}^{\times}} \; \mathcal{L}^{*}({D}^{\ast}_{tr}, w^{\ast});\\
s.t. \; &w^{\ast} \in argmin_w \; \mathcal{L}({D}_{tr}\cup{D}^{\times}).
\end{aligned}
\end{equation}

To poison the locally-trained model of a $DI^{\times}$ to the most extent, the outer optimization injects carefully crafted poisoned data ${D}^{\times}$ into a $DI^{\times}$ to maximize a loss function $\mathcal{L}^{*}$ on an untainted training set ${D}_{tr}^{*}$. The inner optimization minimizes the loss function $\mathcal{L}$ on poisoned data ${D}_{tr}\cup{D}^{\times}$ for local training.

We assume that the total number of DIs in the federated learning is $n$ and the number of malicious DIs injected by the attacker is $n^{\times}=\alpha\cdot n$ ($\alpha<1$).
The attacker controls the fixed fraction $\alpha$ of malicious DIs (\ie $DI^{\times}$).
To withstand the DI-level poisoning attack, the principle goal is to construct a federated model $\textsf{FM}$ with chosen $\textsf{Local}_{ben}^{\ast}$ from encrypted local models $\textsf{Local}$, where $\textsf{Local}_{ben}^{\ast}$ contains all encrypted RF models from benign DIs in the ideal case. However, in the real world, it is difficult to pick all benign RF models from $\textsf{Local}$.

To solve the problem, we use the trimmed loss function to select $\textsf{Local}_{ben}^{\ast}$ with the lowest Mean Squared Error (MSE). The smaller MSE is, the more accurate prediction of RF outputs, where MSE is computed over encrypted validation dataset $[\![{D}_{val}]\!]_{pk}$ that are collected by CP from benign DIs. As MSE contains sensitive information, we compute MSE on an encrypted RF model $[\![rf_i]\!]_{pk_{\text{DI}_i}}$ over the validation dataset $[\![{D}_{val}]\!]_{pk}$ as $[\![{\textsf{MSE}_{rf_i}}]\!]_{pk}$.

To provide an efficient solution against the DI-level poisoning attack, we define the optimization problem in secure defense:
\begin{equation}
\begin{aligned}\label{optimal}
&min\sum_{i=1}^{n^{\ast}}[\![{\textsf{MSE}_{rf_i}}]\!]_{pk}\; \;s.t. \, [\![rf_i]\!]_{pk_{\text{DI}_i}} \in \textsf{Local}_{ben}^{\ast}, \\
\end{aligned}
\end{equation}
where
\begin{equation}
\begin{aligned}\label{MSE}
[\![\textsf{MSE}_{rf_i}]\!]_{pk}&=\sum_{k=1}^{|{D}_{val}|}[\![\textsf{MSE}_{rf_i,d_k}]\!]_{pk},\\
s.t. \,&[\![{d_k}]\!]_{pk} \in [\![{D}_{val}]\!]_{pk}.
\end{aligned}
\end{equation}

We design \textbf{SecureMSE} to securely compute the MSE value $[\![\textsf{MSE}_{rf,d_k}]\!]_{pk}$ over the encrypted RF model $[\![rf]\!]_{pk_{\text{DI}_i}}$ and an encrypted sample $[\![d_k]\!]_{pk}\in D_{val}$. The concrete process is shown as follows.

\paratitle{\textbf{SecureMSE}}$ ([\![rf]\!]_{pk_{\text{DI}_i}},[\![d_k]\!]_{pk})\rightarrow[\![\textsf{MSE}_{rf,d_k}]\!]_{pk} $: At the Step $\textcircled{3}$, given $[\![{d_k}]\!]_{pk} \in[\![{D}_{val}]\!]_{pk}$, the encrypted RF model $[\![rf]\!]_{pk_{\text{DI}_i}}\in \textsf{Local}$ returns the prediction $[\![\hat{y}]\!]_{pk_{\text{DI}_i}}\leftarrow \textbf{SecurePrediction}([\![rf]\!]_{pk_{\text{DI}_i}},[\![{d_k}]\!]_{pk})$. Note that the specific process of prediction is shown in Section~\ref{sec_diagnosis}. After obtaining the prediction $[\![\hat{y}]\!]_{pk_{\text{DI}_i}}$, $[\![{\textsf{MSE}_{rf,d_k}}]\!]_{pk}$ is computed as
\begin{equation}
\begin{aligned}\label{MSE_dk2}
[\![y-\hat{y}]\!]_{pk} &\leftarrow \texttt{SADD}([\![y]\!]_{pk},[\![\hat{y}]\!]_{pk_{\text{DI}_i}}^{N-1}),\\
[\![{\textsf{MSE}_{rf,{d_k}}}]\!]_{pk} &\leftarrow \texttt{SMUL}([\![y-\hat{y}]\!]_{pk},[\![y-\hat{y}]\!]_{pk}),
\end{aligned}
\end{equation}
where $[\![{y}]\!]_{pk}$ is the ground-truth label of $[\![{d_k}]\!]_{pk}$.
Clearly, the corresponding process of $[\![{\textsf{MSE}_{rf,d_k}}\footnote{Note that ${\textsf{MSE}_{rf,d_k}}=(y-\hat{y})^2$.}]\!]_{pk}$ is based on \texttt{SADD} and \texttt{SMUL}, it needs cooperation between CP and the $i$-th DI. Note that $[\![-y]\!]_{pk}=[\![y]\!]_{pk}^{N-1}$ is a homomorphic property~\cite{liu2018efficient}.

Computed with \textbf{SecureMSE}, we can obtain encrypted MSE values $[\![{\textsf{MSE}_{rf_i}}]\!]_{pk}$ $(i=1,2,...,n)$ over a validation dataset $[\![D_{val}]\!]_{pk}$. As shown in step $\textcircled{4}$, we need to remove poisonous models for secure defense, which can implement the solution to the optimal problem in Eq.~\ref{optimal} and provide privacy guarantees at the same time.
An outline of the specific process can be found in \textbf{Algorithm 1}. CP adopts MSE as a discriminator to select suitable local models and construct $\textsf{Local}_{ben}^{\ast}$. If ${\textsf{MSE}_{rf_i}}\geq \Theta$, where $[\![\Theta]\!]_{pk}$ is upper bounder set by the system, we consider the RF model $[\![rf_i]\!]_{pk_{\text{DI}_i}}$ is poisonous and remove it; If ${\textsf{MSE}_{rf_i}}< \Theta$, we consider that $[\![rf_i]\!]_{pk_{\text{DI}_i}}$ is benign. After selecting local models from $\textsf{Local}$, we construct the final federated model with chosen models $\textsf{Local}_{ben}^{\ast}$. Although the subset $\textsf{Local}_{ben}^{\ast}$ still contains the poisoned models, it is close to the benign models and cannot  corrupt the final federated model. Then, the chosen encrypted models $\textsf{Local}_{ben}^{\ast}$ under different public keys can be transformed to $\textsf{FM}$ under the same public key $pk$ via \textbf{SecureAggregation}$(\textsf{Local}_{ben}^{\ast})\rightarrow \textsf{FM}$, where $\textsf{FM}$ contains $n^*$ encrypted RF models. Remark that we recommend a default value of $\Theta=80$ when the size of the validation dataset is $|D_{val}|=100$, and the upper bound $[\![\Theta]\!]_{pk}$ can be adjusted according to the defense's needs and the size of $D_{val}$.

\begin{algorithm}
\scriptsize
\footnotesize
\small
\label{secure_defense}
\caption{Secure Defense}
\KwIn{The validation dataset $D_{val}$, the upper bound $[\![\Theta]\!]_{pk}$ of MSE value, and a set of encrypted locally-trained models $\textsf{Local}=\{[\![rf_i]\!]_{pk_i}\}_{i=1}^{n}$.}
\KwOut{The federated model $\textsf{FM}$.}
Initialize $\textsf{Local}_{ben}^{\ast}\leftarrow\{\}$;\\
\For{\rm each RF model $[\![rf_i]\!]_{pk_i} \in  \textsf{Local}$}{
\For{\rm each data sample $[\![d_k]\!]_{pk}\in D_{val}$}{
$[\![\textsf{MSE}_{rf_i,d_k}]\!]_{pk}\leftarrow \textbf{SecureMSE}([\![rf_i]\!]_{pk_i},[\![d_k]\!]_{pk})$;\\
}
$[\![\textsf{MSE}_{rf_i}]\!]_{pk}=\sum_{k=1}^{|{D}_{val}|}[\![\textsf{MSE}_{rf_i,d_k}]\!]_{pk}$;\\
Secure comparison \texttt{SCOM}$([\![\textsf{MSE}_{rf_i}]\!]_{pk},[\![\Theta]\!]_{pk})$;\\
\If{\rm $\textsf{MSE}_{rf_i}< \Theta$}{
$\textsf{Local}_{ben}^{\ast}\leftarrow \textsf{Local}_{ben}^{\ast} \cup [\![rf_i]\!]_{pk_i}$;\\
}
Construct the federated model $\textsf{FM}$ with $\textsf{Local}_{ben}^{\ast}$ under multiple keys via $\textsf{FM} \leftarrow \textbf{SecureAggregation}(\textsf{Local}_{ben}^{\ast})$;\\
}
\Return $\textsf{FM}$.
\end{algorithm}

\subsection{Privacy-preserving Pocket Diagnosis}\label{sec_diagnosis}

In the pocket diagnosis, an authorized DU can submit his/her symptom characteristics at any time and anywhere for pocket diagnosis service. In the Step $\textcircled{6}$, since symptom characteristics contain DU's privacy, each DU encrypts the submitted characteristics $\{f_1,f_2,...f_{\kappa}\}$ by using his/her public key $pk_u$ as $[\![req]\!]_{pk_u}=\{[\![f_1]\!]_{pk_u},[\![f_2]\!]_{pk_u},...,[\![f_{\kappa}]\!]_{pk_u}\}$.

Then, DU can require the diagnosis service from CP via the submitted-request $[\![req]\!]_{pk_u}$. On receiving $[\![req]\!]_{pk_u}$, CP will adopt the federated model $\textsf{FM}$ to implement privacy-preserving pocket diagnosis. Given $\textsf{FM}$ under the public key $pk$ and the encrypted request $[\![req]\!]_{pk_u}$ under DU's public key $pk_u$, the multi-key secure computation involved in \textbf{SecurePrediction} still needs to be executed between CP and DU as the authorized DU holds the corresponding key share $sk^{(2)}$. The details are shown as follows.

\paratitle{\textbf{SecurePrediction}$([\![rf_i]\!]_{pk_i},[\![req]\!]_{pk_u})$}: At the step $\textcircled{7}$, given an encrypted request $[\![req]\!]_{pk_u}$ submitted by an authorized DU, the federated model $\textsf{FM}=\{[\![rf_i]\!]_{pk}\}_{i=1}^{n*}$ can make a prediction on $[\![req]\!]_{pk_u}$. The prediction on each encrypted tree $[\![dt]\!]_{pk}\in \textsf{FM}$ is defined as $[\![dt]\!]_{pk}([\![req]\!]_{pk_u})$ that involves two keys (\ie $pk$ and $pk_u$). Therefore, the process of pocket diagnosis needs the cooperation between DU and CP for multi-key secure comparison (\textsf{SCOM}). The specific process demonstrates in~\textbf{Algorithm 2}, where each encrypted tree $[\![dt]\!]_{pk}\in \textsf{FM}$ makes the prediction by retrieving the corresponding label from the tree node until encountering a leaf node.

Taking Fig.~\ref{treeStruc} as an example, once receiving an encrypted request $[\![req]\!]_{pk_u}=\{[\![f_1]\!]_{pk_u},[\![f_2]\!]_{pk_u},...,[\![f_{\kappa}]\!]_{pk_u}\}$, the encrypted tree $[\![dt]\!]_{pk}$ will implement prediction as follows:
$[\![dt]\!]_{pk}$ first compares the tree root with $[\![req]\!]_{pk_u}$, as the root is not a leaf, the feature index $id=1$ and split value $[\![w_1]\!]_{pk}$ contained in the weight are obtained, then secure comparison $\texttt{SCOM}([\![w_1]\!]_{pk}, [\![f_1]\!]_{pk_u})$ is called to compare $[\![w_1]\!]_{pk}$ with the corresponding feature value $[\![f_1]\!]_{pk_u} \in [\![req]\!]_{pk_u}$. If $f_1< w_1$, then $[\![req]\!]_{pk_u}$ is compared with the left child, where the feature index $id=2$. Thus, $\texttt{SCOM}([\![w_2]\!]_{pk}, [\![f_2]\!]_{pk_u})$ is implemented to output the comparison result. If $f_2\geq w_2$, then $[\![req]\!]_{pk_u}$ is compared with the right child that is a leaf node. Thus, $c_1$ of the leaf is returned as the prediction result.

After obtaining all predictions from entire encrypted trees contained in $\textsf{FM}$, the final diagnosis result $[\![\hat{y}_{dia}]\!]_{pk}$ is provided as
\begin{equation}
\begin{aligned}\label{diagnosis}
[\![\hat{y}_{{dia}}]\!]_{pk}=\frac{1}{t}\sum_{i=1}^{n^{\ast}}\sum_{j=1}^{t_j}[\![dt_j]\!]_{pk}([\![x]\!]_{pk});\\
s.t. \; [\![dt_j]\!]_{pk}\in \textsf{FM},\; t=t_1+t_2+...+t_{n^{\ast}}.
\end{aligned}
\end{equation}

\begin{figure}
  \centering
  \includegraphics[width=2.5in]{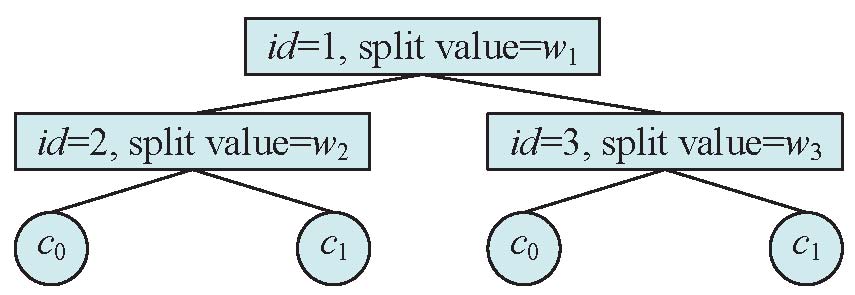}\\
  \caption{The structure of a tree.}\label{treeStruc}
\end{figure}

\begin{algorithm}
\scriptsize
\footnotesize
\small
\label{rfPrediction}
\caption{Secure Prediction on Encrypted Tree}
\KwIn{Encrypted instance $[\![req]\!]_{pk_u}=\{[\![f_1]\!]_{pk_u},...,[\![f_{\kappa}]\!]_{pk_u}\}$, and encrypted tree $[\![dt]\!]_{pk}$.}
\KwOut{Encrypted prediction result ${[\![\hat{y}]\!]}_{pk}$.}
Initialize the result $[\![\hat{y}]\!]_{pk}\leftarrow \textbf{Enc}_{pk}(0)$; \\
Initialize the feature index $id$;\\
$[\![node]\!]_{pk}\leftarrow$ the tree root of $[\![dt]\!]_{pk}$;\\
\While{true}{
\If{$node$ is a leaf}
{Obtain label $[\![c]\!]_{pk}$ of the leaf;\\
\Return ${[\![\hat{y}]\!]}_{pk}\leftarrow[\![c]\!]_{pk}$;\\
}
\Else{
Obtain the weight $(id,[\![w]\!]_{pk})$ from $[\![node]\!]_{pk}$;\\
Obtain the selected feature value $[\![f_{id}]\!]_{pk_u}\in[\![req]\!]_{pk_u}$;\\
Secure comparison \texttt{SCOM}$([\![f_{id}]\!]_{pk_u},[\![w]\!]_{pk})$;\\
\If{$f_{id} < w$}{
$[\![node]\!]_{pk} \leftarrow [\![node]\!]_{pk}.leftChild$;
}
\Else{$[\![node]\!]_{pk}\leftarrow [\![node]\!]_{pk}.rightChild$;}
}
}
\Return ${[\![\hat{y}]\!]}_{pk}$.
\end{algorithm}
To protect the privacy of DU's diagnosis result, $[\![\hat{y}_{dia}]\!]_{pk}$ can only be known to DU. In the step $\textcircled{8}$, the diagnosis result is securely returned according to following operations.
\begin{itemize}
  \item  CP first computes $ [\![\hat{y}_{{dia}}]\!]_{pk}^{(1)}\leftarrow \texttt{SDec}_{sk^{(1)}}([\![\hat{y}_{{dia}}]\!]_{pk})$, then sends both $[\![\hat{y}_{{dia}}]\!]_{pk}$ and $[\![\hat{y}_{{dia}}]\!]_{pk}^{(1)}$ to DU.
  \item On receiving the encrypted result, DU calculates $[\![\hat{y}_{{dia}}]\!]_{pk}^{(2)}\leftarrow \texttt{SDec}_{sk^{(2)}}([\![\hat{y}_{{dia}}]\!]_{pk})$ to obtain the decryption share, and then decrypts the final diagnosis result via
      $\hat{y}_{{dia}}\leftarrow \texttt{WDec}([\![\hat{y}_{{dia}}]\!]_{pk}^{(1)},[\![\hat{y}_{{dia}}]\!]_{pk}^{(2)})$ that is only know to DU.
\end{itemize}

\textbf{Remarks}. SFPA not only provides a more accurate diagnosis service by using the privacy-preserving RF-based federated learning in the ideal world, but also designs an efficient defense countermeasure against poisoning attacks in the adversarial world. Compared with previous schemes, our SFPA also supports the multi-key setting. Besides, SFPA framework can be extended with XGBoost, linear regression and deep learning techniques.
\section{Privacy Analysis}\label{sec:privacyAnal}
Here, we define the \textit{real vs. ideal} model to formalize the privacy analysis in SFPA. Specifically, assume that an adversary $\mathcal{A}$ interacts with a challenger in the real world to perform the algorithm $\prod$, then it interacts with a simulator $\mathcal{S}$ to complete the process in the ideal world. If the view of $\mathcal{A}$ in the real world is indistinguishable from the view of $\mathcal{S}$ in the ideal world, then we consider the algorithm $\prod$ is secure, as represented in $\{\text{IDEAL}_{\mathcal{S}}\}\stackrel{c}{\equiv} \{\text{REAL}_{\mathcal{A}}\}$,
where the symbol $\stackrel{c}{\equiv}$ means computational indistinguishability.
Besides, we analyze the security of our SFPA system as follows.
\begin{theorem}\rm
Our privacy-preserving RF-based federated learning is secure against semi-honest CP or DIs as long as multi-key secure computation achieves the security~\cite{liu2016efficient}, which can resist the distinguishment of intermediate computation results if there exists no collusion between CP and DIs.
\end{theorem}

\begin{proof}
In the following, we separately analysis the security in our federated learning with the above \textit{real vs. ideal} model.

\textit{CP-side Security}: During the process of secure aggregation, the view of $\mathcal{A}_{\text{CP}}$ is defined as $\text{REAL}_{\mathcal{A}_\text{CP}}=(\textsf{Local}, \{sk^{(1)}_{\text{DI}_i}\}_{i=1}^{n},\textsf{Local}^{(1)},\textsf{FM})$, where $\textsf{Local}=\{pk,[\![rf_i]\!]_{pk_{\text{DI}_i}}\}_{i=1}^{n}$ are encrypted local models, $\{sk^{(1)}_{\text{DI}_i}\}_{i=1}^{n}$ are secret shares of corresponding $\{sk_{\text{DI}_i}\}_{i=1}^{n}$, and their decryption shares $\textsf{Local}^{(1)}=\{[\![rf_i]\!]_{pk_{\text{DI}_i}}^{(1)}\}_{i=1}^{n}$.
Given the tuple $(\textsf{Local}, \{sk_{\text{DI}_i}^{(1)}\}_{i=1}^{n})$, we construct the simulator $\mathcal{S}_{\text{CP}}$ as follows:
 \begin{itemize}
   \item CP blinds random numbers to obtain $\textsf{Local}'$, and use \texttt{SDec} to get decryption shares $\small{{\overline{\textsf{Local}}'^{(1)}}}=\{[\![\overline{rf_i'}]\!]_{pk_{\text{DI}_i}}^{(1)}\}_{i=1}^n$.
   \item The $i$-th DI operates $\texttt{SDec}_{sk_{\text{DI}_i}^{(2)}}(\cdot)$ and \texttt{WDec} to get the corresponding plaintexts $\pi$ of $\textsf{Local}'$ that contain blinded numbers, then returns $\small\overline{\textsf{FM}}\leftarrow \texttt{Enc}_{pk}(\cdot)$ to $\mathcal{S}_{\text{CP}}$.
 \end{itemize}

Obliviously, the distributions of the real-world $\text{REAL}_{\mathcal{A}_\text{CP}}=(\textsf{Local}, \{sk^{(1)}_{\text{DI}_i}\}_{i=1}^{n},\textsf{Local}^{(1)},\textsf{FM})$ and the distributions of the ideal-world $\text{IDEAL}_{\mathcal{S}_\text{CP}}=(\textsf{Local}, \{sk^{(1)}_{\text{DI}_i}\}_{i=1}^{n}$,${\overline{\textsf{Local}}'^{(1)}},\overline{\textsf{FM}})$ are indistinguishable due to the semantic security of the proved public-key cryptosystem, as demonstrated in $\{\text{IDEAL}_{\mathcal{S}_{\text{CP}}}\}\stackrel{c}{\equiv} \{\text{REAL}_{\mathcal{A}_{\text{CP}}}\}$.

\textit{DI-side Security}: Similarly, the view of $\mathcal{A}_{\text{DI}}$ is defined as $\text{REAL}_{\mathcal{A}_{\text{DI}}}=({\textsf{Local}'^{(1)}},{\textsf{Local}'},{sk_{\text{DI}_i}^{(2)}},\pi)$, which includes the intermediate numbers $\pi$ blinded by random numbers. Thus, given $({\textsf{Local}'^{(1)}},{\textsf{Local}'},{sk_{\text{DI}_i}^{(2)}})$ to construct ${\mathcal{S}_{\text{DI}}}$ with the above operations. The view of ${\mathcal{S}_{\text{DI}}}$ is $\text{IDEAL}_{\mathcal{S}_{\text{DI}}}=({\textsf{Local}'^{(1)}},{\textsf{Local}'},{sk_{\text{DI}_i}^{(2)}},\overline{\pi})$, it is unable to deduce any privacy from the intermediate operations due to blinded random numbers.
Therefore, we can conclude that $\{\text{IDEAL}_{\mathcal{S}_{\text{DI}}}\}\stackrel{c}{\equiv} \{\text{REAL}_{\mathcal{A}_{\text{DI}}}\}.$

Here, the secure aggregation is proved to be secure against semi-honest adversaries $\mathcal{A}_{\text{CP}}$ and $\mathcal{A}_{\text{DI}}$ those can corrupt CP and DIs in the real world, respectively.
\end{proof}

\begin{theorem}\rm
Our defense countermeasure which deals with DI-level poisoning attacks is secure against semi-honest DI or CP on condition that the multi-key secure computation is secure and there exist no collusion attacks between CP and DIs.
\end{theorem}

\begin{proof}
Similar to \textbf{Theorem 1}, the specific analysis is demonstrated as follows.

\textit{CP-side Security}: During the process of secure defense, the view of $\mathcal{A}_{\text{CP}}$ is defined as $\text{REAL}_{\mathcal{A}_\text{CP}}=(\textsf{Local}, \{sk^{(1)}_{\text{DI}_i}\}_{i=1}^{n},\textsf{Local}^{(1)},\textsf{FM})$. To implement the \textbf{SecureMSE} and secure defense algorithm, subsequent interactions in secure computation (\ie \texttt{SADD}, \texttt{STRA}, \texttt{SCOM}, \texttt{SMUL}) between CP and DIs are involved with multiple keys. Although CP holds the partial secret shares $\{sk^{(1)}_{\text{DI}_i}\}_{i=1}^{n}$, $sk^{(1)}$ and each DI holds $\{sk^{(2)}_{\text{DI}_i}\}_{i=1}^{n}$, $sk^{(2)}$, none of them can obtain the plaintexts as the blinding technique is employed in the secure computation. The view of ${\mathcal{S}_{\text{CP}}}$ is $\text{REAL}_{\mathcal{A}_\text{CP}}=(\textsf{Local}, \{sk^{(1)}_{\text{DI}_i}\}_{i=1}^{n},{\overline{\textsf{Local}}'^{(1)}},\overline{\textsf{FM}})$. Obviously, it is computational indistinguishability, which is represented as $\{\text{IDEAL}_{\mathcal{S}_{\text{CP}}}\}\stackrel{c}{\equiv} \{\text{REAL}_{\mathcal{A}_{\text{CP}}}\}.$

\textit{DI-side Security}: Similarly, the distribution of views between ${\mathcal{A}_{\text{DI}}}$ and ${\mathcal{S}_{\text{DI}}}$ are also identical, which are shown as $\{\text{IDEAL}_{\mathcal{S}_{\text{DI}}}\}\stackrel{c}{\equiv} \{\text{REAL}_{\mathcal{A}_{\text{DI}}}\}.$

For more proof details, please refer to the reference~\cite{liu2016efficient}. Both the adversary $\mathcal{A}_{\text{CP}}$ corrupting CP and the adversary $\mathcal{A}_{\text{DI}}$ corrupting DIs are unable to distinguish the ideal world from the real world.
\end{proof}

\begin{theorem}\rm
Our pocket diagnosis is secure against semi-honest CP or DU on condition that the multi-key secure computation can realize the distinguishment of intermediate computation results and there exist no collusion attacks between CP and DIs.
\end{theorem}

\begin{proof}\rm
The specific analysis is shown as follows.

\textit{CP-side Security}: During the process of pocket diagnosis, the view of ${\mathcal{A}_\text{CP}}$ is defines as $\text{REAL}_{\mathcal{A}_\text{CP}}=([\![req]\!]_{pk_u}, sk^{(1)}_{u},\textsf{FM})$, where $[\![req]\!]_{pk_u}$ represents an encrypted request sent by DU for diagnosis service, $sk^{(1)}_{u}$ denotes DU's corresponding secret share. Given $([\![req]\!]_{pk_u}, sk^{(1)}_{u},\textsf{FM})$, we adopt multi-key secure computation to construct the simulator ${\mathcal{S}_\text{CP}}$, the process interacts with DU and provides the proved privacy protection~\cite{liu2016efficient}. Therefore, it is impossible to infer the private information, the distribution of views between $\text{REAL}_{\mathcal{A}_\text{CP}}$ and $\text{IDEAL}_{\mathcal{S}_\text{CP}}$ is indistinguishable, as represented in $\{\text{IDEAL}_{\mathcal{S}_{\text{CP}}}\}\stackrel{c}{\equiv} \{\text{REAL}_{\mathcal{A}_{\text{CP}}}\}.$

\textit{DU-side Security}:
In addition, the returned diagnosis result $[\![\hat{y}_{dia}]\!]_{pk}^{(2)}$ is encrypted under public key $pk$, it is worth noticing that DU holds the secret share $sk^{(2)}$ to decrypt the final diagnosis results. Even if the adversary steals the federated model, it still cannot use it without the corresponding secret key $sk^{(2)}$. Thus, we conclude that $\{\text{IDEAL}_{\mathcal{S}_{\text{DU}}}\}\stackrel{c}{\equiv} \{\text{REAL}_{\mathcal{A}_{\text{DU}}}\}.$

Therefore, both the adversary $\mathcal{A}_{\text{CP}}$ corrupting CP and the adversary $\mathcal{A}_{\text{DU}}$ corrupting DU in the real world are unable to distinguish.
\end{proof}

To guarantee the security of the SFPA system, the malicious adversaries cannot distinguish encrypted information from the other encrypted data without the corresponding secret keys. Therefore, SFPA system can resist privacy leakage with the above theorems.

\section{Performance Analysis}\label{sec:perf}

In this section, we first clarify two real-world datasets and the experiment setup, then evaluate the performance of SFPA by comparing it with other typical privacy-preserving RF-based machine learning schemes~\cite{ma2019privacy}. Finally, we validate the effectiveness of the secure defense against DI-level poisoning attacks.

\subsection{Dataset and Experiment Setup}

We perform our evaluation on two public datasets:

\begin{enumerate}
    \item[(a)] \textbf{Heart disease dataset}\footnote{\url{http://archive.ics.uci.edu/ml/datasets/Heart+Disease}} includes 303 instances, 14 features as well as two labels, where a patient without heart disease is labeled as ``-1", a patient suffers from heart disease ``1".
  \item[(b)] \textbf{Thyroid disease dataset}\footnote{\url{http://www.kaggle.com/kumar012/hypothyroid}} contains 3,163 instances, 18 features and two labels, where a patient without thyroid disease is annotated as ``-1", a patient with thyroid disease is annotated as ``1".
\end{enumerate}

\begin{figure}[!ht]
	  \centering
      \subfigure[Feature distribution of 3 sub-datasets belonging to ${D}_{tr_a}$.]{\includegraphics[width=1.65in]{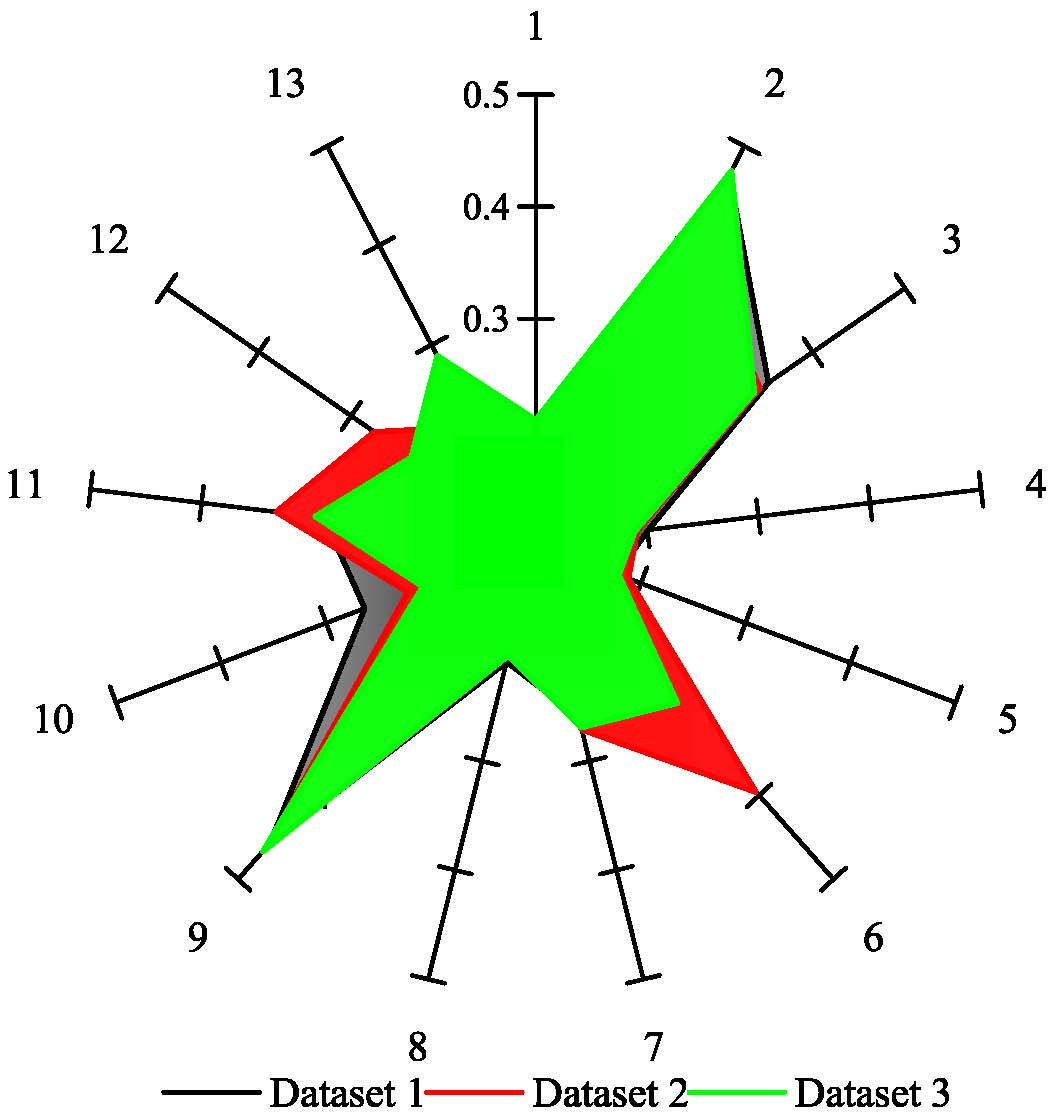}}
      \subfigure[Feature distribution of 10 sub-datasets belonging to ${D}_{tr_b}$.]{\includegraphics[width=1.8in]{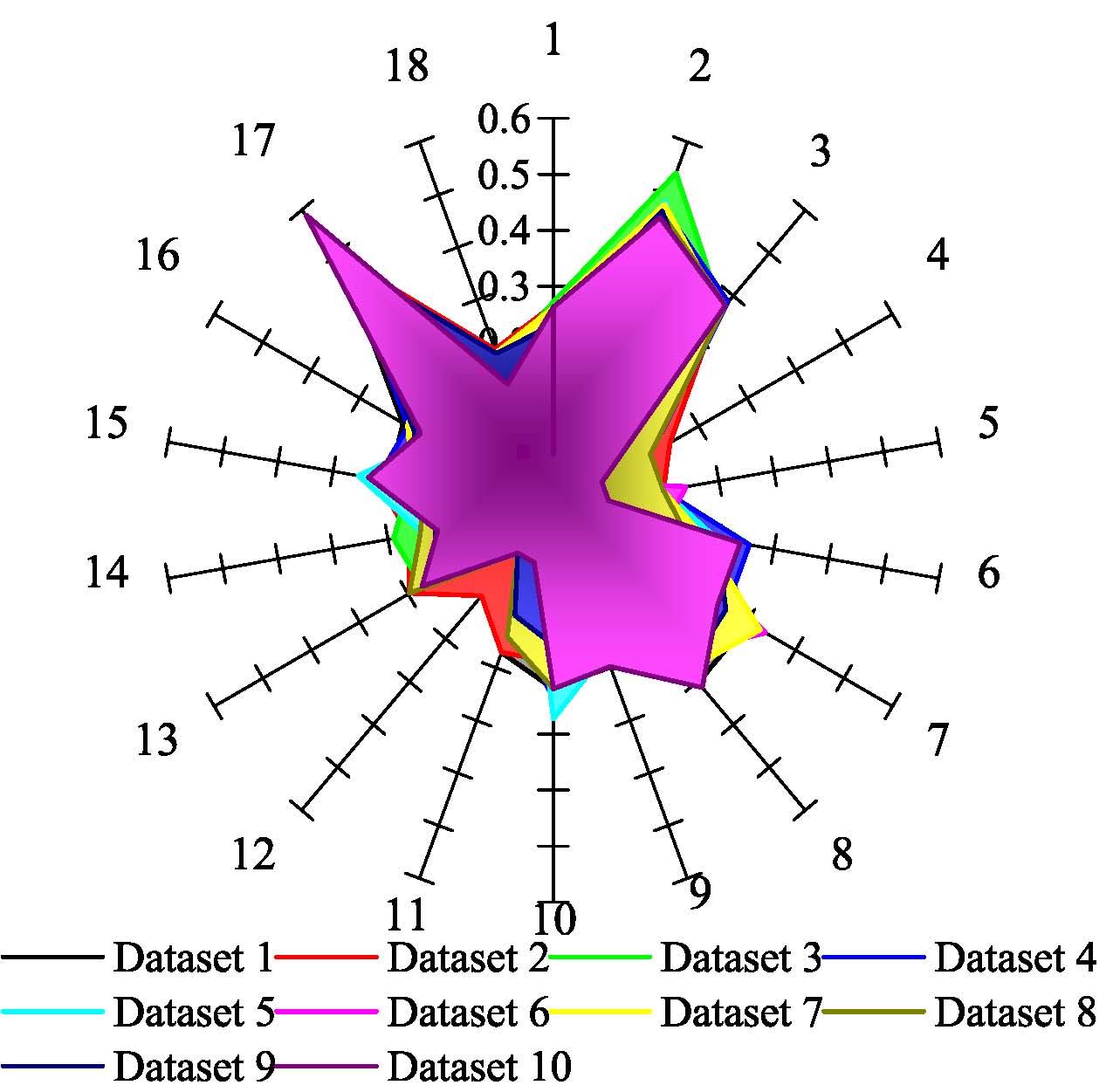}}
      \caption{Non-i.i.d distribution of training datasets.} \label{distribution}
\end{figure}

We use the cross-validation method to split training datasets into the training set ($\frac{2}{3}$), validation and testing set ($\frac{1}{3}$), heart disease dataset (abbr., ${D}_{tr_a}$) is divided into 3 sub-datasets and thyroid disease dataset (abbr., ${D}_{tr_b}$) is divided into 10 sub-datasets for federated learning, where the size of each sub-dataset is 100 samples. In Figs.~\ref{distribution}(a)(b), it shows that Standard Deviation (StdDev) of each feature contained in all sub-datasets of ${D}_{tr_a}$ and ${D}_{tr_b}$, respectively, where StdDev is denoted as the degree of dispersion of samples in a dataset. As presented as Figs.~\ref{distribution}(a)(b), we discover that the non-i.i.d. distribution of all sub-datasets of ${D}_{tr_a}$ and ${D}_{tr_b}$, as there are bias in the feature distributions among sub-datasets.

In our experiments, we choose $|N|=1024$ bits in the multi-key decryption scheme {shown in Section~\ref{PublicKeyCrypto}} to achieve 80-bit security level. Besides, we implemented our SFPA in Java, the experiments are evaluated on the cloud platform (25 VMs from the GRNET Okeanos public cloud\footnote{\url{https://okeanos.grnet.gr/}} with the following characteristics per VM: 1 VCPU, 1 GB RAM, 10 GB Disk, Ubuntu Server 12.04.4 LTS).


\subsection{Performance of Privacy-preserving RF-based Federated Learning}

In the experiment, we first plot the setting of federated learning with 3 DIs in the ideal world. Since the tree number $t$ is an important factor in each RF, we achieve a tradeoff between effectiveness and efficiency by choosing different values of $t$.

\begin{figure}[!ht]
	  \centering
      \subfigure[Federated learning (3 DIs, ${D}_{tr_a}$).]{\includegraphics[width=1.15in]{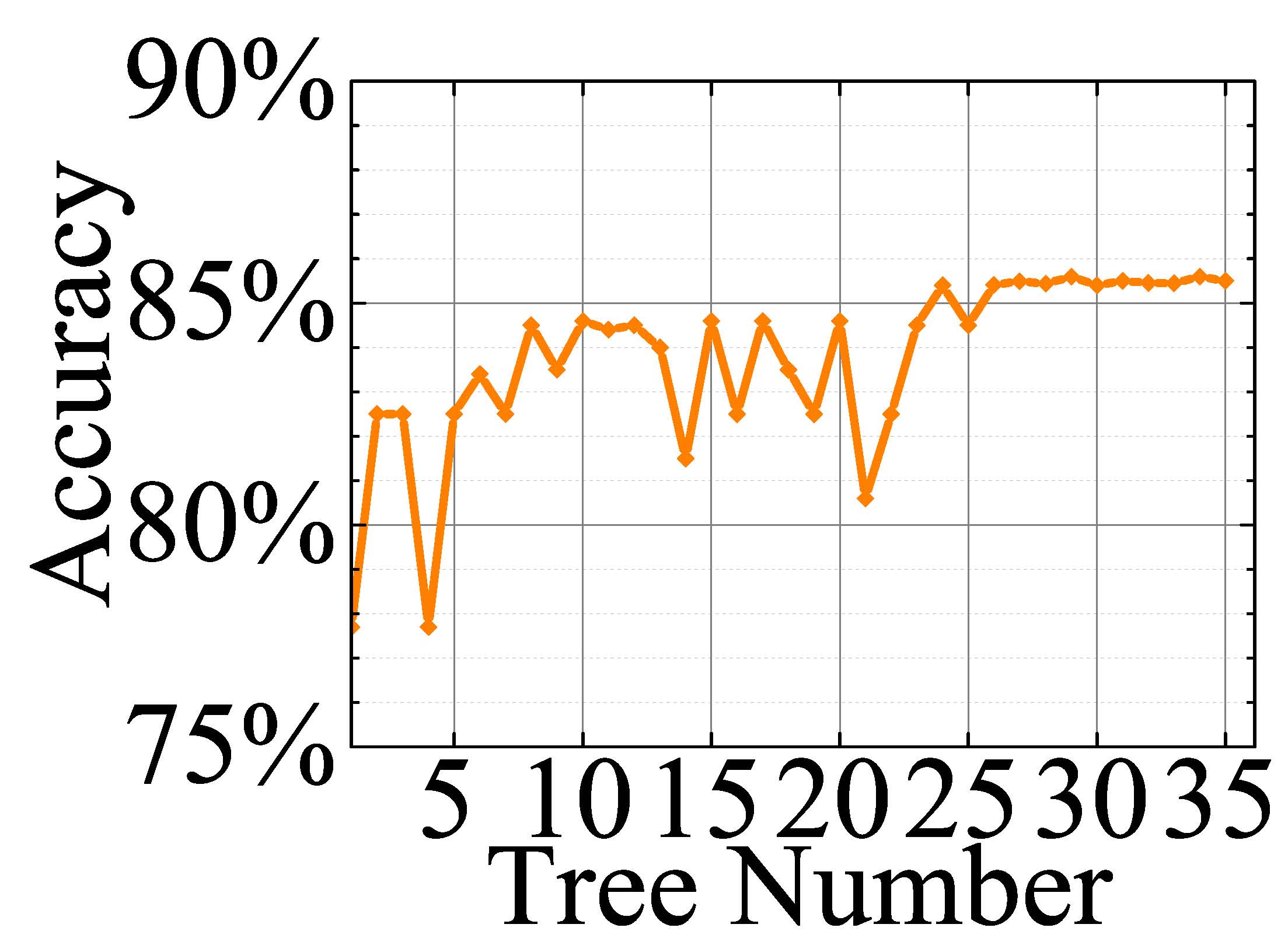}}
      \subfigure[Federated learning (3 DIs, ${D}_{tr_b}$).]{\includegraphics[width=1.15in]{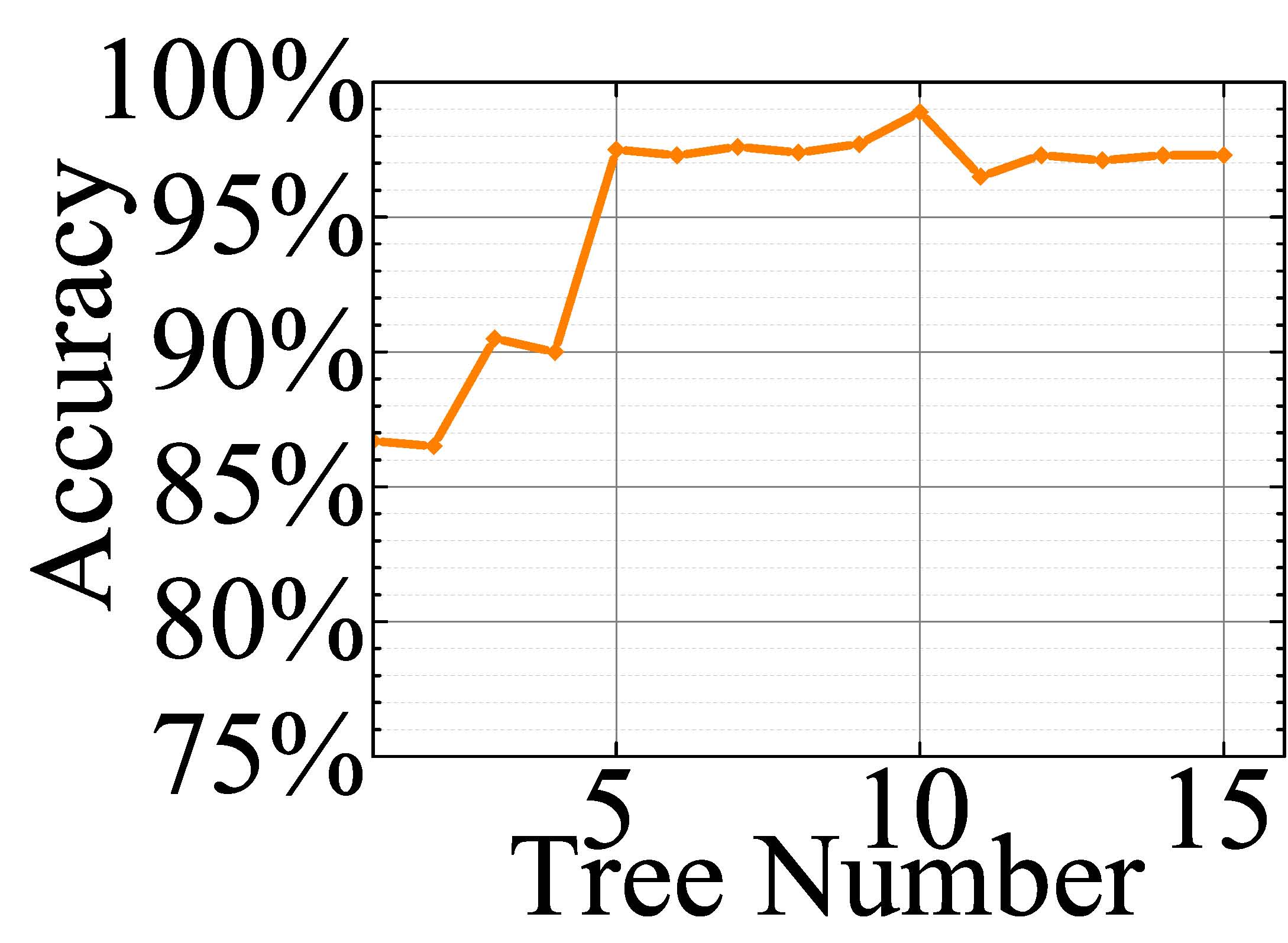}}
      \subfigure[Secure aggregation (3 DIs, ${D}_{tr_a}$).]{\includegraphics[width=1.1in]{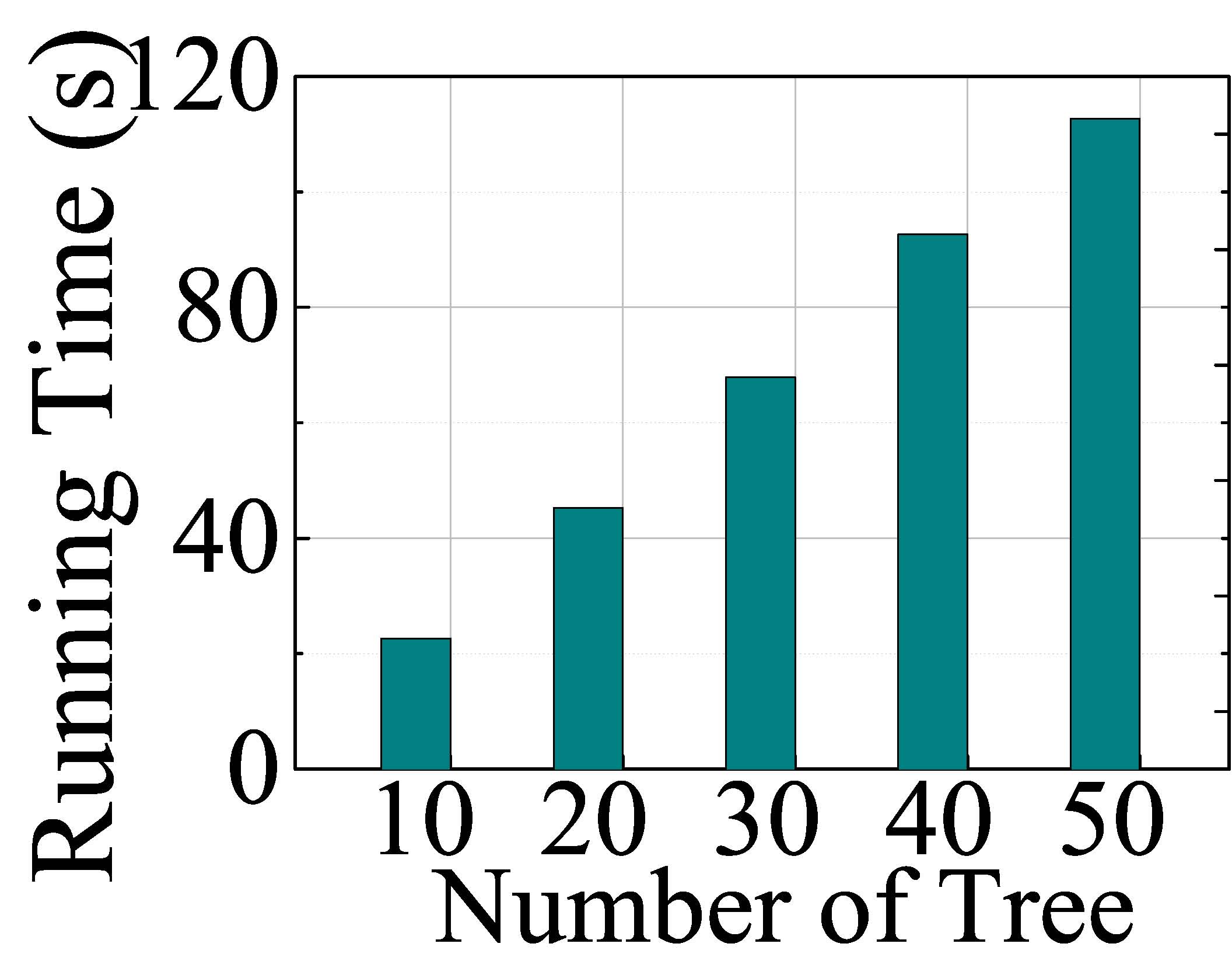}}

      \caption{Performance of SFPA. (a)(b) are the accuracy of the proposed federated learning varying with tree number $t$, where 3 DIs and 10 DIs, respectively. (c) is the running time of secure aggregation varying with $t$, where 3 DIs.} \label{flperformance}
\end{figure}

\subsubsection{Effectiveness Analysis}
As discussed in Figs.~\ref{flperformance}(a)(b), we observe that the greater the value of $t$ is, the higher accuracy outputs. Besides, the accuracy increases slowly and remains stable over ${D}_{tr_a}$ when $t>30$, and the accuracy is $85.6\%$ when $t=30$. Fig.~\ref{flperformance}(b) shows that the accuracy slowly increases remains stable over ${D}_{tr_b}$ when $t>5$, and the accuracy is $97.5\%$ when $t=5$. In Fig.~\ref{flperformance}(c), we notice that the running time of secure aggregation has approximately linear growth with the tree number $t$ over ${D}_{tr_a}$. This is because with the growth of tree number, more trees are required to be aggregated. Thus, it needs more \texttt{STRA} involved to implement secure aggregation. Considering both efficiency and accuracy, we respectively set $t=30$ of ${D}_{tr_a}$, and $t=5$ of ${D}_{tr_b}$ in the following experiments.

\subsubsection{Efficiency Analysis}
Then, we evaluate the efficiency in our implementation over the training dataset ${D}_{tr_a}$:
\begin{itemize}
  \item \textit{DI-side processing}: It costs 0.110 s for locally training in each DI, where each DI adopts a RF model that consists of 30 trees.
  \item \textit{CP-side processing}: It costs 9.84 s for secure aggregation in our federated learning.
\end{itemize}

\subsubsection{Comparison Analysis}
A comparison of both computational overhead (abbr., Comp.) and communication overhead (abbr., Comm.) among our proposed federated learning, Ma \textit{et al.}~\cite{ma2019privacy}, and Aono \textit{et al.}~\cite{aono2017privacy} is conducted in TABLE~\ref{diagnosis_comparedfl}. Both~\cite{ma2019privacy,aono2017privacy} are based on Paillier cryptosystem ($|N|=1024$ bits). In~\cite{aono2017privacy}, the Convolutional Neural Network (CNN) model consists of two convolution layers and two fully connected layers with 128 neurons each layer, the activation function is ReLU function.
As plotted in TABLE~\ref{diagnosis_comparedfl}, the comparison analysis is shown in the following.

\begin{figure}[!ht]
	  \centering
      \subfigure[Comparison between~\cite{ma2019privacy} and SFPA (3 DIs, ${D}_{tr_a}$, $t=1$).]{\includegraphics[width=1.77in]{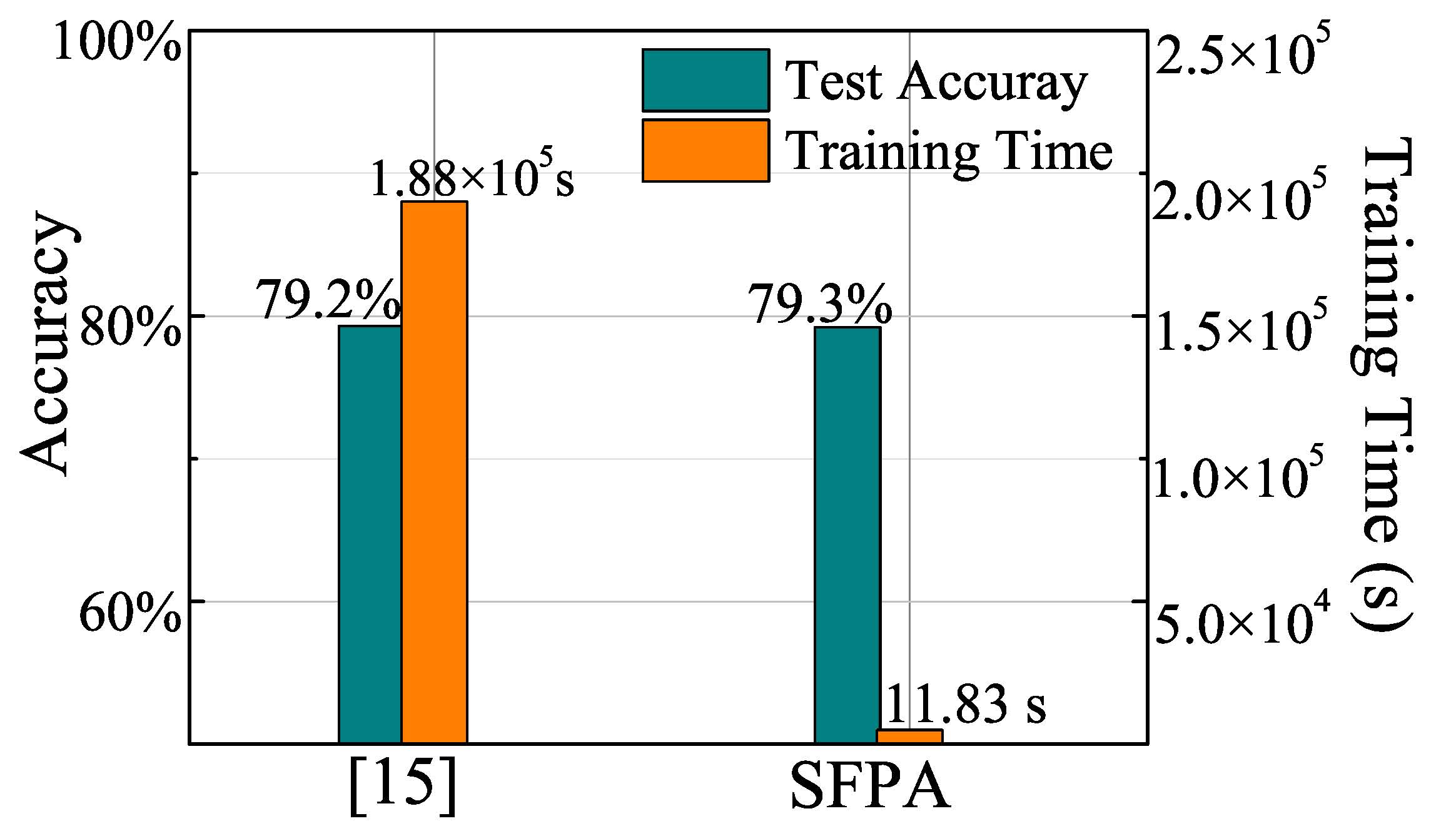}}
      \subfigure[Comparison between~\cite{aono2017privacy} and SFPA (3 DIs, ${D}_{tr_b}$, $t=5$).]{\includegraphics[width=1.68in]{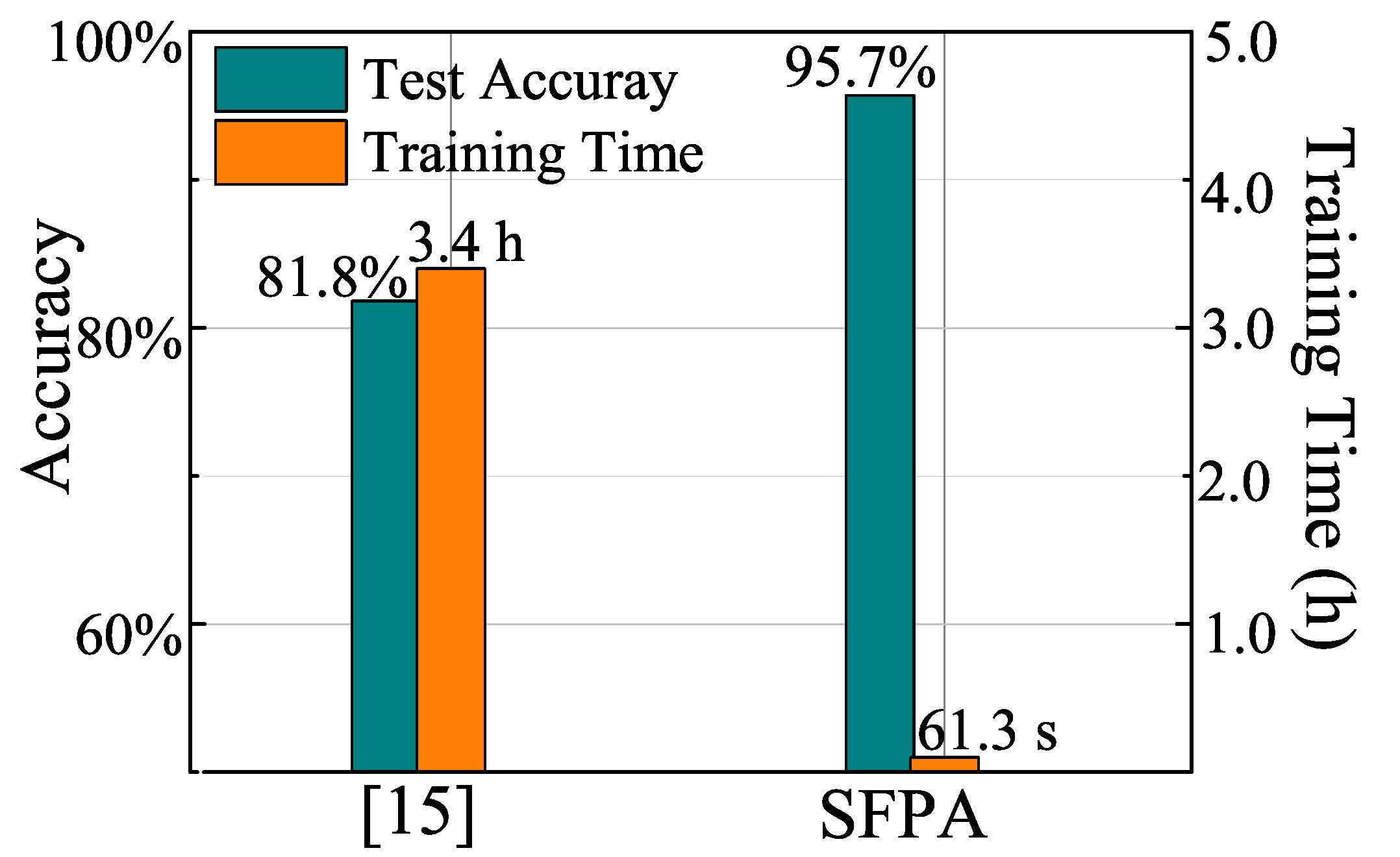}}

      \caption{Comparison analysis of SFPA.} \label{flcomparison}
\end{figure}

{Compared with privacy-preserving data-outsourced random forest scheme~\cite{ma2019privacy}}, we discover that our scheme has a significant improvement of training time and a negligible decrease of communication overhead. The reason is that our method implements local training over plaintexts instead of running over ciphertexts, which obviously reduces ciphertext operations to improve efficiency.

In the actual dataset ${D}_{tr_a}$, we set the tree number $t=1$ to run our federated learning and \cite{ma2019privacy}. As demonstrated in Fig.~\ref{flcomparison}(a), the training time is 11.831 s in our method while the time is $1.88\times 10^5$ s in \cite{ma2019privacy}, and the accuracy is $79.3\%$ in our method while $79.2\%$ in~\cite{ma2019privacy}. Hence, we can conclude that our method brings in an obvious increase of efficiency without a decrease in accuracy.

{Compared with privacy-preserving federated learning based on deep learning~\cite{aono2017privacy}}, we discover that SFPA has a significant improvement of training time and communication overhead. As a CNN model contains hundreds of neurons in~\cite{aono2017privacy}, which not only costs a lot of computational overhead to implement homomorphic addition operations, but also burdens the communication overhead in transmitted CNN models between CP and a DI.

In the actual dataset $D_{tr_b}$, we set the tree number $t=5$ and 3 DIs to run SFPA and~\cite{aono2017privacy}. As demonstrated in Fig.~\ref{flcomparison}(b), the training time is 61.317 s in our method while the time is $3.4$ hours in~\cite{aono2017privacy}, and the accuracy is $95.7\%$ in our method while $81.8\%$ in~\cite{aono2017privacy}. Hence, we can conclude that SFPA brings in an obvious increase in accuracy, since the excellent prediction of random forest over textual data. Besides, the obvious improvement of efficiency is because that compared with a large number of neurons contained in a CNN model of~\cite{aono2017privacy}, a random forest is a lightweight model contained tens of tree nodes, which can significantly improve computational overhead as fewer nodes are involved in secure aggregation.

\subsection{Performance of Secure Defense and Pocket Diagnosis}

\begin{figure*}[!ht]
	  \centering
      \subfigure[3 DIs, $t=30$, ${D}_{tr_a}$.]{\includegraphics[width=1.67in]{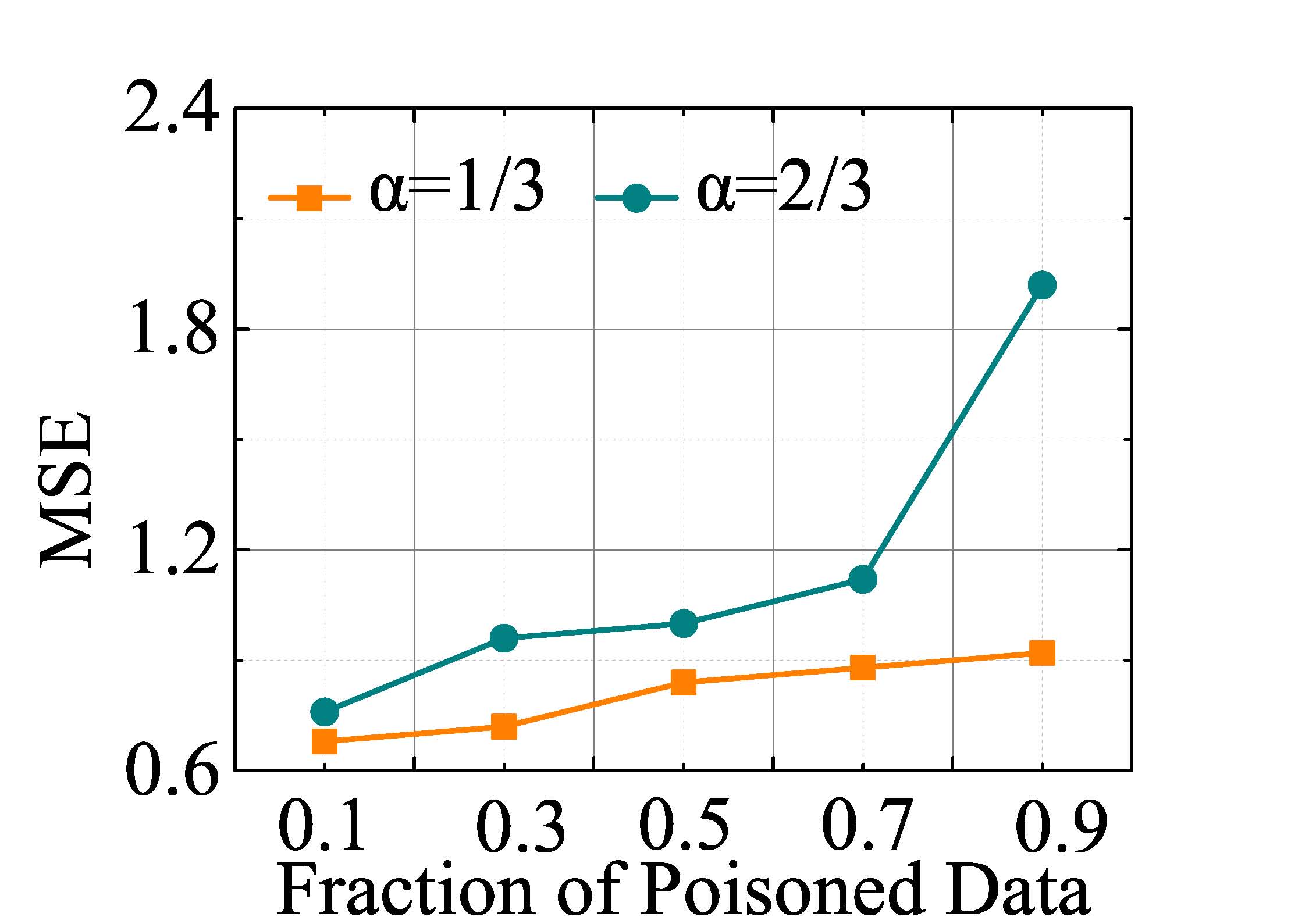}}
      \subfigure[3 DIs, $t=30$, $\alpha=1/3$, ${D}_{tr_a}$.]{\includegraphics[width=1.7in]{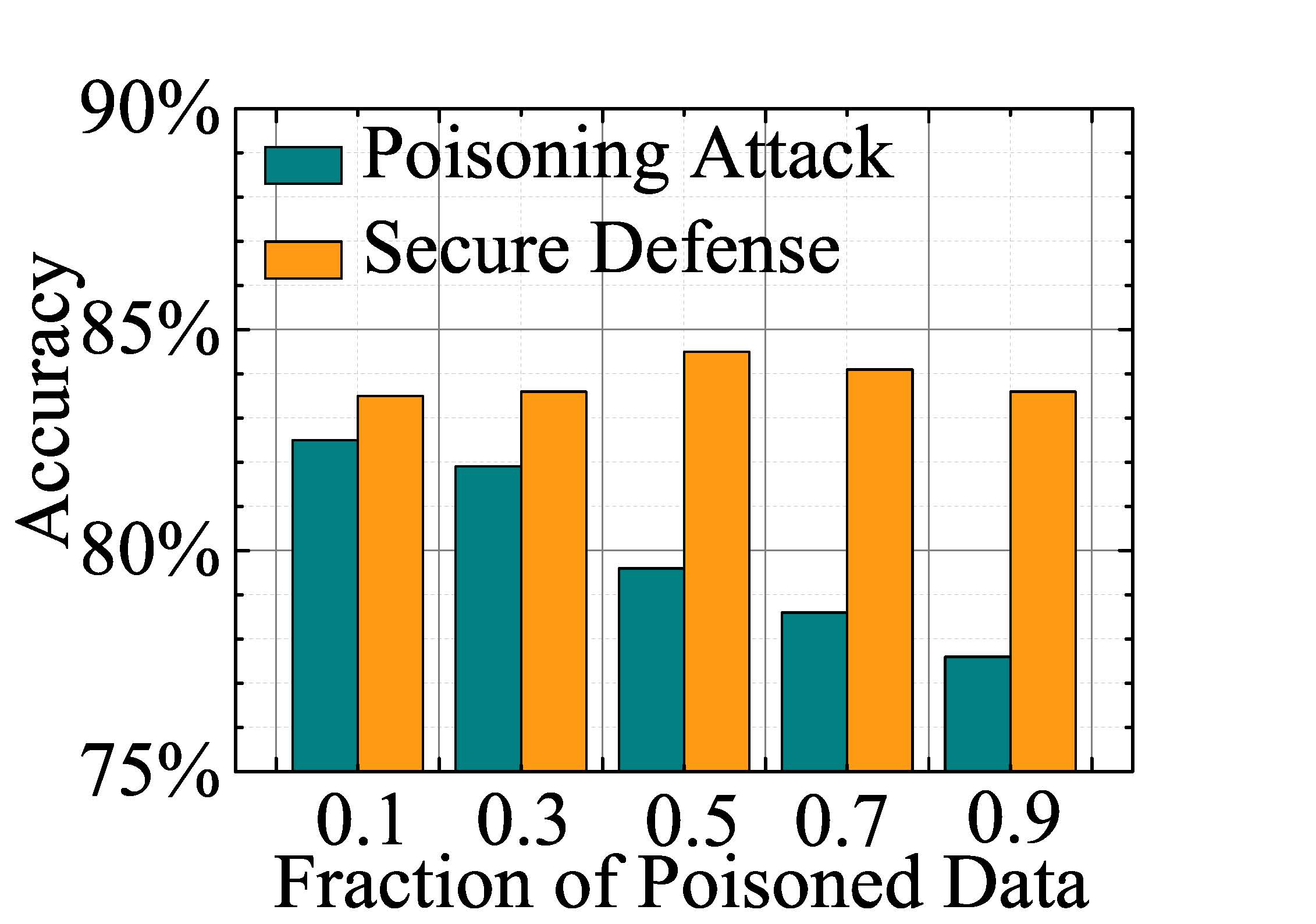}}
      \subfigure[3 DIs, $t=30$, $\alpha=2/3$, ${D}_{tr_a}$.]{\includegraphics[width=1.7in]{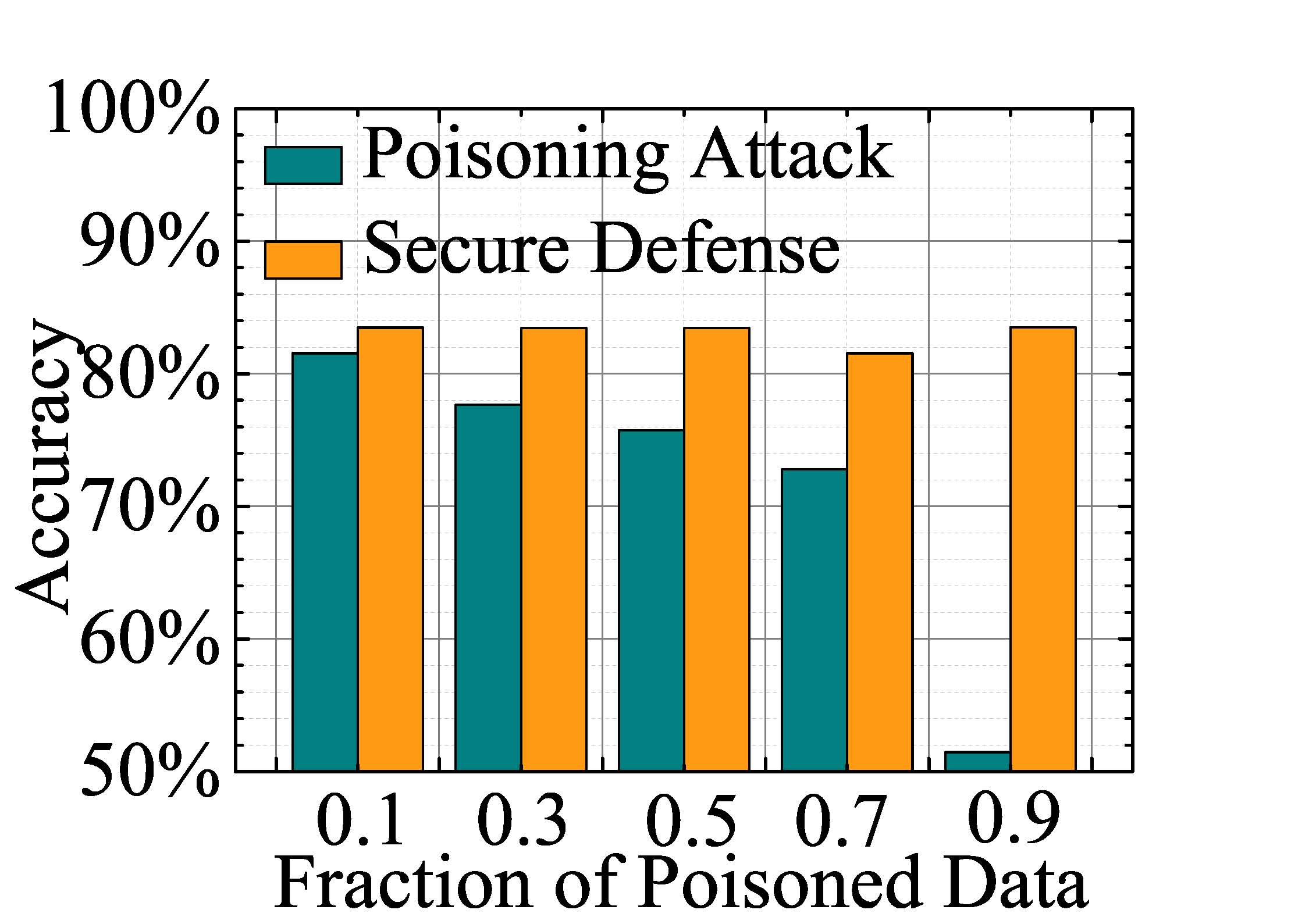}}
     \subfigure[Accuracy with different poisoning attacks, 10 DIs, $t=5$, ${D}_{tr_b}$.]{\includegraphics[width=1.8in]{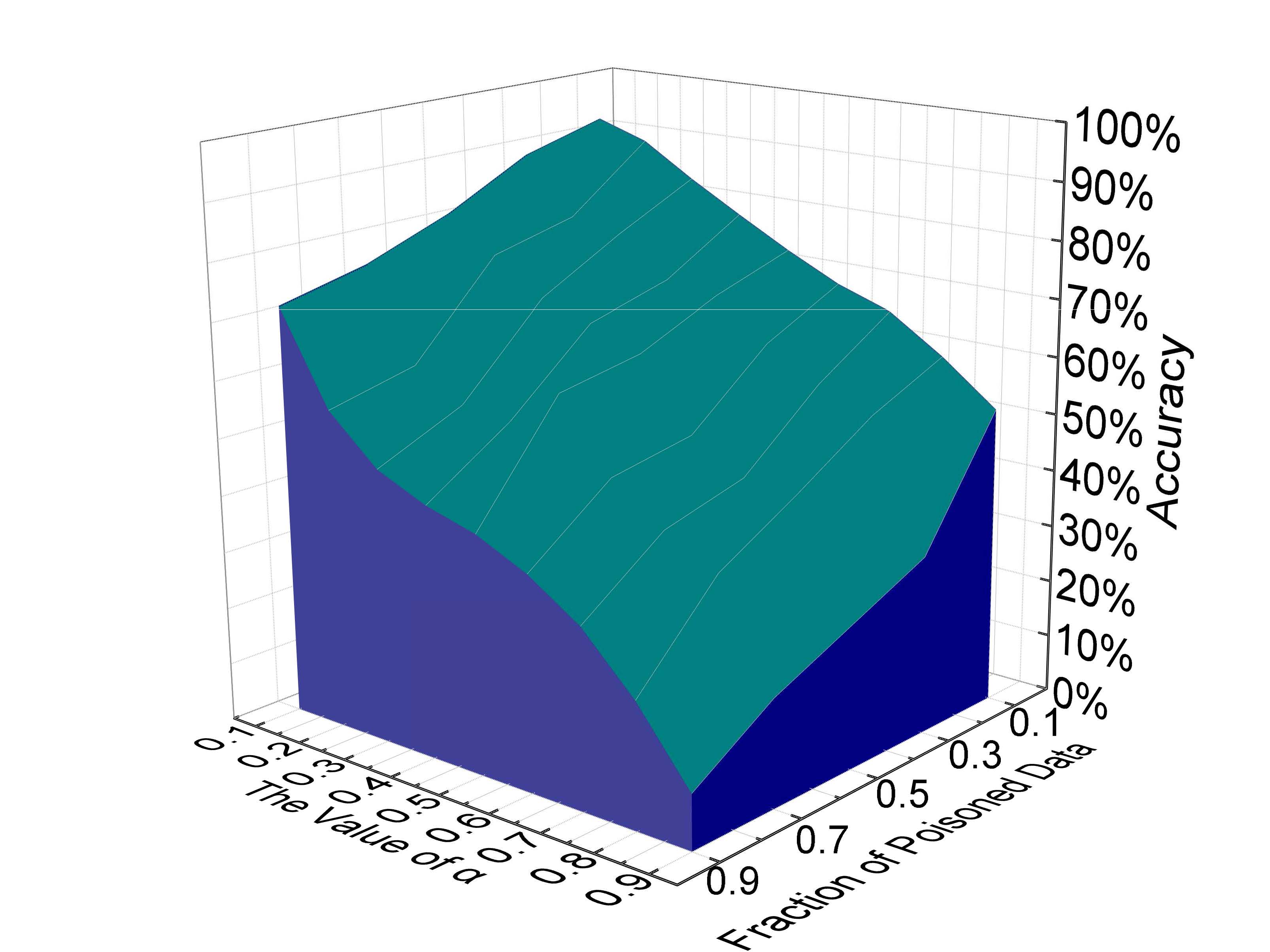}}
     \subfigure[Accuracy with defense strategy, 10 DIs, $t=5$, ${D}_{tr_b}$.]{\includegraphics[width=1.8in]{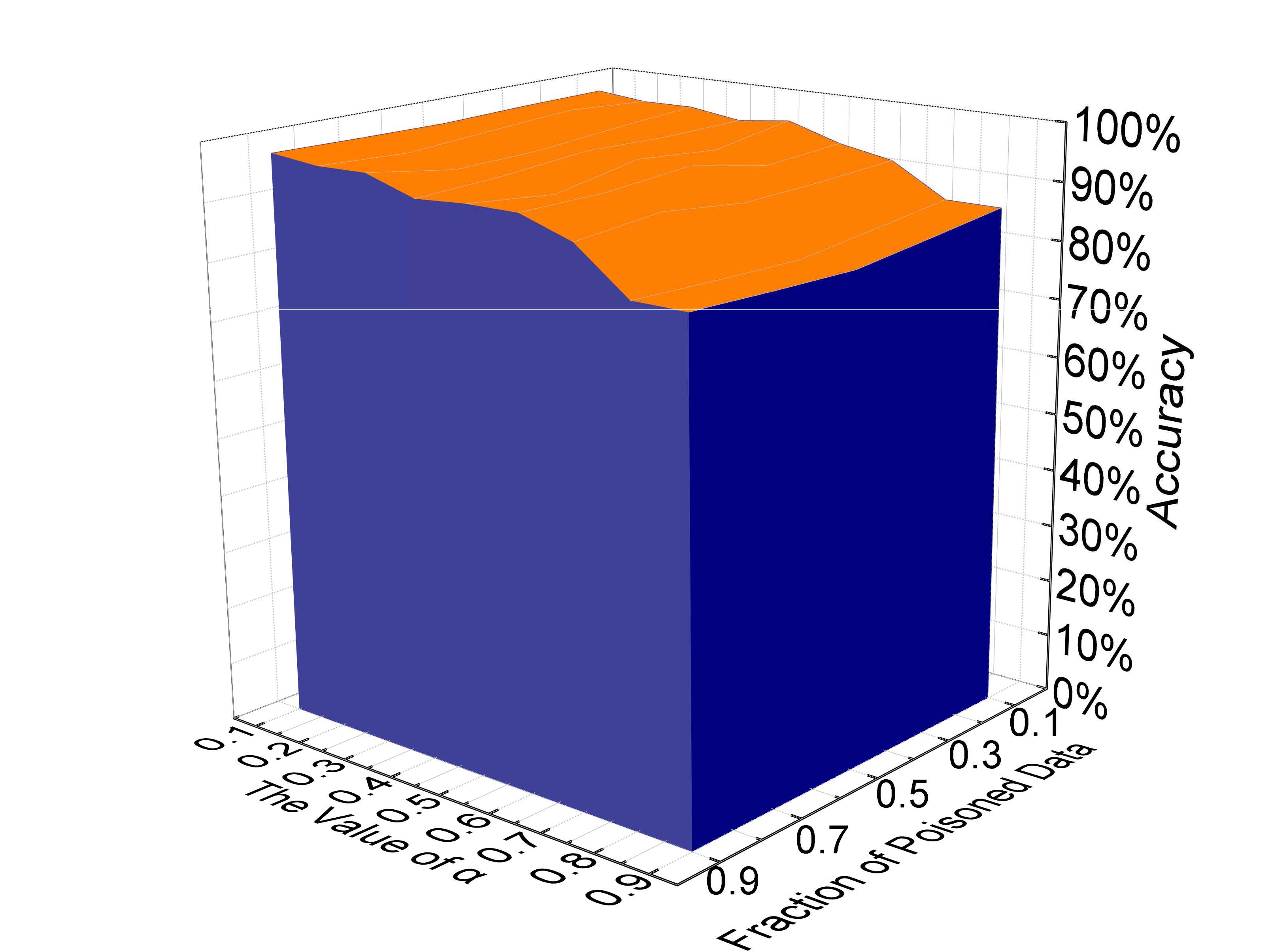}}
       \subfigure[$\beta=0.9$, $\alpha=1/3$, ${D}_{tr_a}$.]{\includegraphics[width=1.75in]{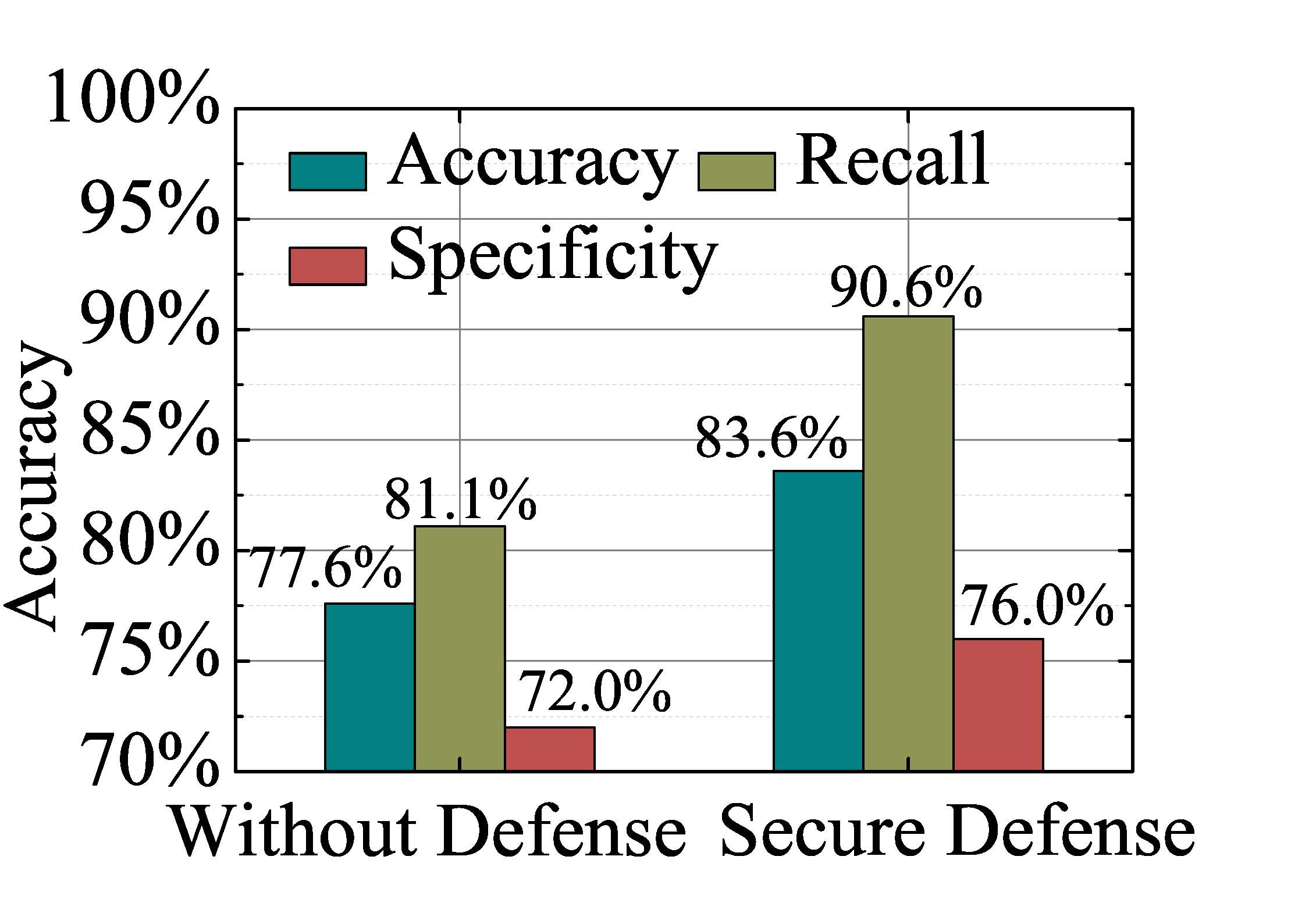}}
       \subfigure[$t=30$, $\alpha=1/3$, ${D}_{tr_a}$.]{\includegraphics[width=1.75in]{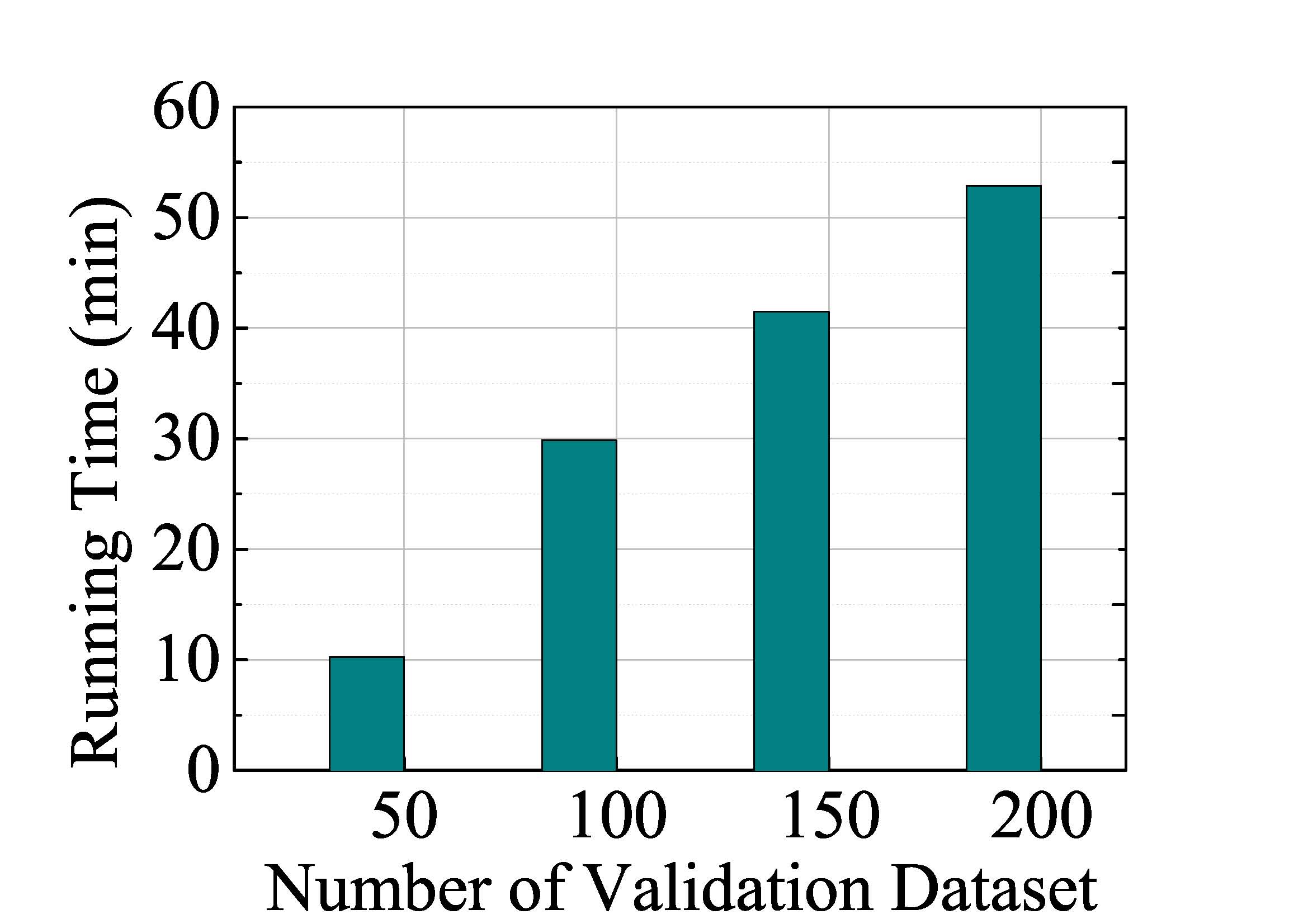}}
 \subfigure[ {$|D_{val}|$=100, $\alpha=1/3$}, ${D}_{tr_a}$.]{\includegraphics[width=1.75in]{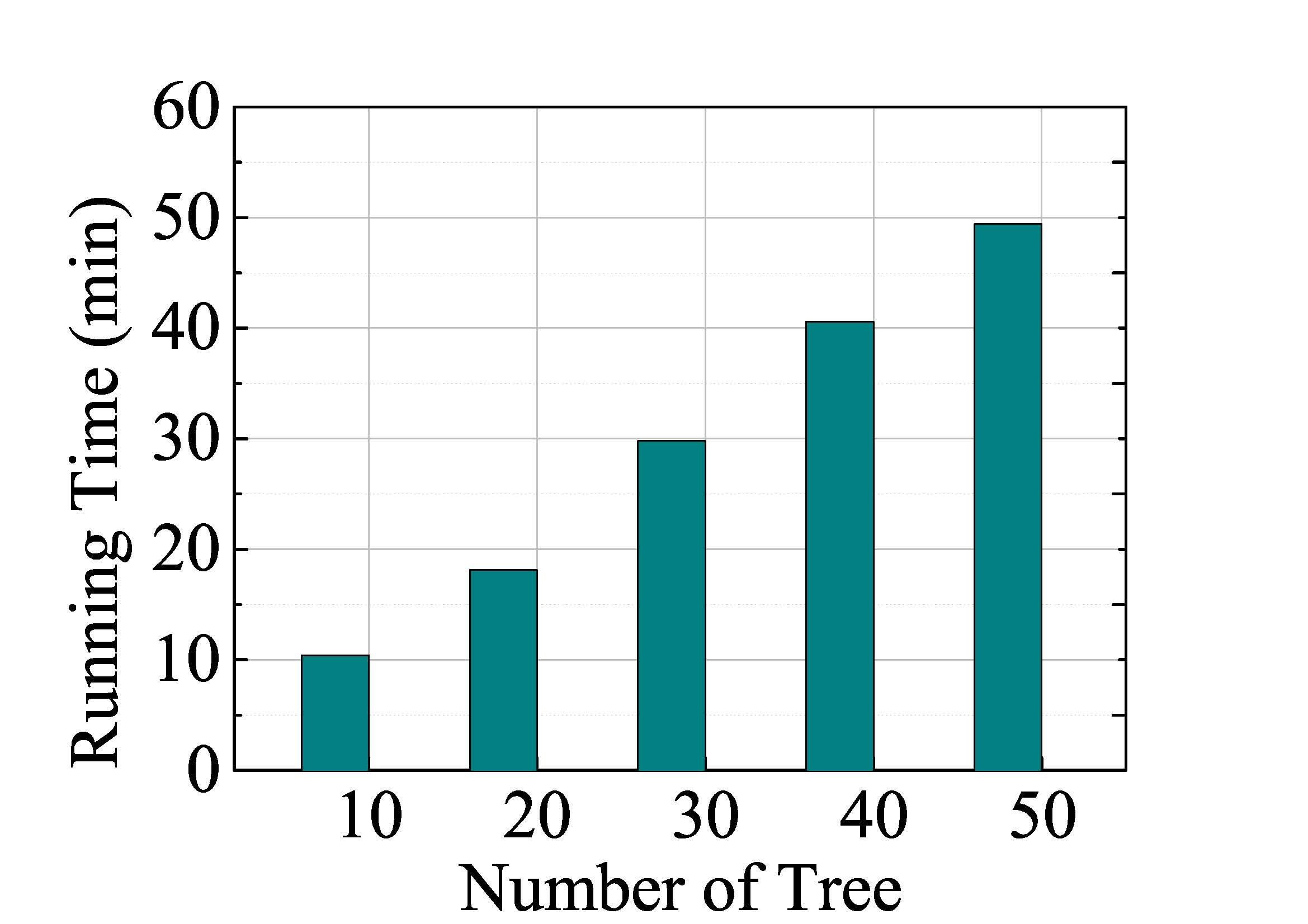}}

      \caption{Performance of secure defense countermeasure. (a) is the MSE of DI-level poisoning attack with different fraction, where 3 DIs and $t=30$. (b)(c) are the comparison of accuracy between secure defense and the proposed federated learning without defense, where the $\alpha=1/3$ and $\alpha=2/3$, respectively. (d) is the accuracy of different $\alpha$ and different fraction of poisoned data. (e) is the accuracy of secure defense with different $\alpha$ and different fraction of poisoned data. (f) is the comparison between secure defense and without defense, where $\beta=0.9$ and $\alpha=1/3$. (g)(h) are running time of secure defense in different situations.} \label{performance}
\end{figure*}

\begin{table*}
\centering
  \scriptsize
\caption{Comparison between our federated learning and ~\cite{ma2019privacy,aono2017privacy} }\label{diagnosis_comparedfl}
 \tabcolsep 21pt 
\begin{tabular*}{7in}{l|c||c||c||c}
\toprule
  \hline
\multicolumn{2}{c||}{Phase} &Our framework &Ma \textit{et al.}~\cite{ma2019privacy} &Aono \textit{et al.}~\cite{aono2017privacy}\\
  \hline
 \multirow{2}{*}{Comp.}& Init. & $|D_{tr}|vh+n(2^h-1){\texttt{T}_\texttt{Enc}}$  & $|D_{tr}| {\texttt{T}_\texttt{Enc}}$&$|Neurons|{\texttt{T}_\texttt{Enc}}$\\
{}& Training&$nt(2^h-1){\texttt{T}_\texttt{STAR}}$&$t(2^h-1)\texttt{T}_{\texttt{total}}$&$n|Neurons| {\texttt{T}_\texttt{Add}}$ \\
  \hline
\hline

\multirow{2}{*}{Comm.}&Init. & $nt(2^h-1)\Im$&$|D_{tr}| \Im$&$|Neurons| \Im$\\
{}& Training&$3n(2^h-1)\Im$&$3t(2^h-1)(\epsilon_1  +\epsilon_2+\epsilon_3) \Im$ &-- \\
  \hline
\bottomrule
 \end{tabular*}
 \begin{tablenotes}
\item \textbf{Notes}. $|D_{tr}|$ is the total number of training data hosted over all DIs, $\Im$ is the bit length of an encrypted number, $v$ is the number of features, $h$ is the height of a tree, $n$ is the total number of DIs. $\texttt{T}_{\texttt{Enc}}$ and $\texttt{T}_{\texttt{STRA}}$ denote the time for \texttt{Enc} and \texttt{STRA}, respectively. Besides, $\texttt{T}_{\texttt{total}}=\epsilon_1  {\texttt{T}_\texttt{Enc}}+\epsilon_2{\texttt{T}_\texttt{Add}+\epsilon_3{\texttt{T}_\texttt{Mul}}}$ are the total computation overhead of secure computation in~\cite{ma2019privacy}, where $\texttt{T}_\texttt{Add}$, ${\texttt{T}_\texttt{Mul}}$ are the time for homomorphic addition and homomorphic multiplication and $\epsilon_1$, $\epsilon_2$, $\epsilon_3$ are operation times of \texttt{Enc}, homomorphic addition and homomorphic multiplication in~\cite{ma2019privacy}, respectively. $|Neurons|$ denotes the number of neurons of a CNN in~\cite{aono2017privacy}. ``--" means no communication overhead is required in training process of~\cite{aono2017privacy}, as the whole process is based on additive homomorphic.
\end{tablenotes}
\end{table*}

We perform the proposed secure defense on DI-level poisoning attacks. Meanwhile, the proposed federated learning is also constructed in the adversarial world. In order to compare the effectiveness of the secure defense, we plot the adversarial setting of federated learning with 3 sub-datasets belonging to ${D}_{tr_a}$ under the setting of 3 DIs, where the poisoning rate is set as $\alpha=1/3$ (\ie one poisoned ${DI}^{\times}$) and $\alpha=2/3$ (\ie two poisoned ${DIs}^{\times}$), and the adversarial setting of federated learning with 10 sub-datasets belonging to ${D}_{tr_b}$ under the setting of 10 DIs, where the poisoning rate is set as $\alpha\in(0,1)$, respectively.
To deploy the DI-level attack, we randomly clone the fraction $\beta$\footnote{To avoid the detection of DI, the attack injects malicious samples within a limited time, thus the fraction $\beta$ is difficult to reach 1.} ($0<\beta<1$) of the local training data of $DI^{\times}$ and flip their labels~\cite{Biggio2012Poisoning} as the attacker's carefully crafted poisonous data injected into $DI^{\times}$'s training data. It is worth noticing that the new label $y^{\times}$ of each poisonous data is obtained by $y^{\times}=0-y$, where $y\in[-1,1]$ is the original data label. We vary the fraction $\beta$ of $DI^{\times}$ from 0.1 to 0.9 with intervals of 0.2 for implementation of DI-level poisoning attack.
\begin{table}
\centering
  \scriptsize
\caption{Computational overhead and communication overhead }\label{analyCom}
\tabcolsep 5pt
\begin{tabular*}{3.5in}{l|c|c}
  \toprule
   \hline
{Phase}  & Computational Overhead &Communication Overhead \\
 \hline
 {ModelBuild}&$\mathcal{O}(vh|D_{tr}|)$ &\multirow{2}{*}{$\mathcal{O}(nt2^h\Im)$} \\
{{ModelEnc}}&$\mathcal{O}(nt2^h\texttt{T}_{\texttt{Enc}})$\\
 \hline
{SecureMSE} &$\mathcal{O}(nth|D_{val}|\texttt{T}_{\texttt{SCOM}})$ &$\mathcal{O}(3nt2^h\Im|D_{val}|)$ \\
{SecureAggregation}&$\mathcal{O}(n^*t2^h\texttt{T}_{\texttt{STRA}})$&$\mathcal{O}(3n^*t2^h\Im)$\\
 \hline
{SecurePrediction}  &$\mathcal{O}(n^*th\texttt{T}_{\texttt{SCOM}})$ &$\mathcal{O}(3n^*th\Im)$ \\
 \hline
 \bottomrule
 \end{tabular*}
 \begin{tablenotes}
\item \textbf{Notes}. $\Im$ is the bit length of an encrypted number, $|D_{tr}|$ is the total number of training data hosted over all DIs, $v$ is the number of features, $h$ is the height of a tree, $n$ is the total number of DIs, and $n^*$ is the number of chosen RF models for federated model. $\texttt{T}_{\texttt{Enc}}$, $\texttt{T}_{\texttt{SCOM}}$ and $\texttt{T}_{\texttt{STRA}}$ denote the time for \texttt{Enc}, \texttt{SCOM} and \texttt{STRA}, respectively.
\end{tablenotes}
\end{table}

\subsubsection{Effectiveness Analysis}
To observe the impact of DI-level poisoning attacks on the federated model over ${D}_{tr_a}$ under the setting of $\alpha=1/3$ and $\alpha=2/3$, we first implement the proposed federated learning over 3 DIs without secure defense. According to the MSE of DI-level poisoning attack shown in Fig.~\ref{performance}(a), we highlight several interesting observations. First, $\beta$ has a positive impact on poisoning the local RF. The larger the value of $\beta$ is, the greater MSE products (the lower accuracy local poisoned RF model trains). Second, with the greater $\alpha$, the bigger MSE outputs. Hence, $\alpha$ also has an important effect on the DI-level poisoning attack.

In the experiments presented below, we evaluate the newly proposed secure defense in our SFPA. The accuracy of the DI-level poisoning attack varying with different fraction of poisonous data is compared with that of secure defense in Figs.~\ref{performance}(b)-(e). As plotted in Figs.~\ref{performance}(b)(c), our secure defense countermeasure can efficiently resist the data poisoning-level attack over the training dataset $D_{tr_a}$, which has a significant improvement compared with the proposed federated learning without defense when $\alpha=1/3$ and $\alpha=2/3$, and works well in the stable accuracy range ($83\%$-$84.5\%$). Besides, the accuracy of SFPA with secure defense does not be influenced by the $\alpha$.

{As plotted in Fig.~\ref{performance}(d), we discover that the accuracy of our proposed federated learning is going down over the training dataset $D_{tr_b}$ with the growth of $\alpha$ and the increase of the fraction of poisoned data. As plotted in Fig.~\ref{performance}(e), we discover that our secure defense countermeasure can efficiently resist the data poisoning-level attack over $D_{tr_b}$ compared with Fig.~\ref{performance}(d) within the stable accuracy range ($84.3\%$-$97.4\%$). }
Besides, in Fig.~\ref{performance}(f), we introduce ``recall" to measure the federated model's capability of classifying positive samples, and employs ``specificity" to measure the federated model's capability of classifying negative samples. When $\beta=0.9$ and $\alpha=1/3$, the recall, specificity and accuracy have obvious improvement in secure defense over training data $D_{tr_a}$ when compared with the federated learning without defense.

\subsubsection{Efficiency Analysis}
As demonstrated in Figs.~\ref{performance}(g)(h), we evaluate the computation cost of the proposed secure defense over training dataset $D_{tr_a}$ with setting $\alpha=1/3$ by varying two main factors, namely tree number $t$ and the size of validation dataset $|D_{val}|$.

\textit{CP-side processing}: As demonstrated in Fig.~\ref{performance}(g), we discover that the running time of secure defense linearly increases with varying the number of validation dataset $|D_{val}|$ from 50 to 200. This is the fact that the larger number of validation data brings in more computation costs caused by calling \textbf{SecureMSE} algorithm. Besides, from Fig.~\ref{performance}(h), we observe that the running time increases when changing the tree number $t$ from 10 to 50. This is because more encrypted trees are involved in the process of \textbf{SecureMSE}, which leads to more secure predictions on an encrypted tree. From the above figures, we discover that the running time of secure defense linearly increases with tree number and the size of validation dataset.

When $|D_{val}|=100$, $\beta=0.5$ and $t=30$, for the CP-side processing, secure defense costs 1793.3 s and the accuracy of pocket diagnosis is stable as $84.5\%$ (recall = $86.8\%$, specificity = $82.0\%$). For the {DU-side processing}, an authorized DU can make a pocket diagnosis in 5.84 s except for the communication overhead.

\subsection{Cost Analysis}

Furthermore, we analyze the complexities in SFPA system, the computational overhead and communication overhead are detailedly listed in TABLE~\ref{analyCom}.
\begin{itemize}
  \item \textit{DI-side processing}: During the process of initialization, it involves $\textbf{ModelBuild}$ and $\textbf{ModelEnc}$, and costs the total time complexity as $\mathcal{O}(t2^h\texttt{T}_{\text{Enc}})$ for local training over multiple DIs.
  \item \textit{CP-side processing}: During the process of secure defense, it involves $\textbf{SecureMSE}$ and $\textbf{SecureAggregation}$, and spends the total time complexity as $\mathcal{O}(nth|D_{val}|\texttt{T}_{\texttt{SCOM}}+n^*t2^h\texttt{T}_{\texttt{STRA}})$ to construct the federated model.
  \item \textit{DU-side processing}: For the diagnosis phase, it involves $\textbf{SecurePrediction}$ to make a diagnosis for a DU's encrypted request, which costs $\mathcal{O}(n^*th\texttt{T}_{\texttt{SCOM}})$ for a diagnosis.
\end{itemize}

\section{Conclusion}\label{sec:conclusion}
This paper has proposed the privacy-preserving random forest-based federated learning with secure defense against DI-level poisoning attacks, which can achieve real-time and accurate pocket diagnosis. According to multi-key secure computation, the proposed SFPA could not only securely construct federated learning based on random forest, but also efficiently defense poisoning attacks without comprising the availability of the federated model. Experimental results over real-world datasets verified the efficiency and security of the SFPA in the federated learning setting.

\section*{Acknowledgments}
This work was supported by the Key Program of NSFC (No. U1405255), the Shaanxi Science \& Technology Coordination \& Innovation Project (No. 2016TZC-G-6-3), the National Natural Science Foundation of China (No. 61702404, No. 61702105, No. U1804263), the Fundamental Research Funds for the Central Universities (No. JB171504, No. JB191506), the National Natural Science Foundation of Shaanxi Province (No. 2019JQ-005), Guangxi Key Laboratory of Cryptography and Information Security (No. GCIS201917).

\bibliographystyle{IEEEtran}
\bibliography{sensordata}

\begin{IEEEbiography}[{\includegraphics[width=1in,height=1.25in,clip,keepaspectratio]{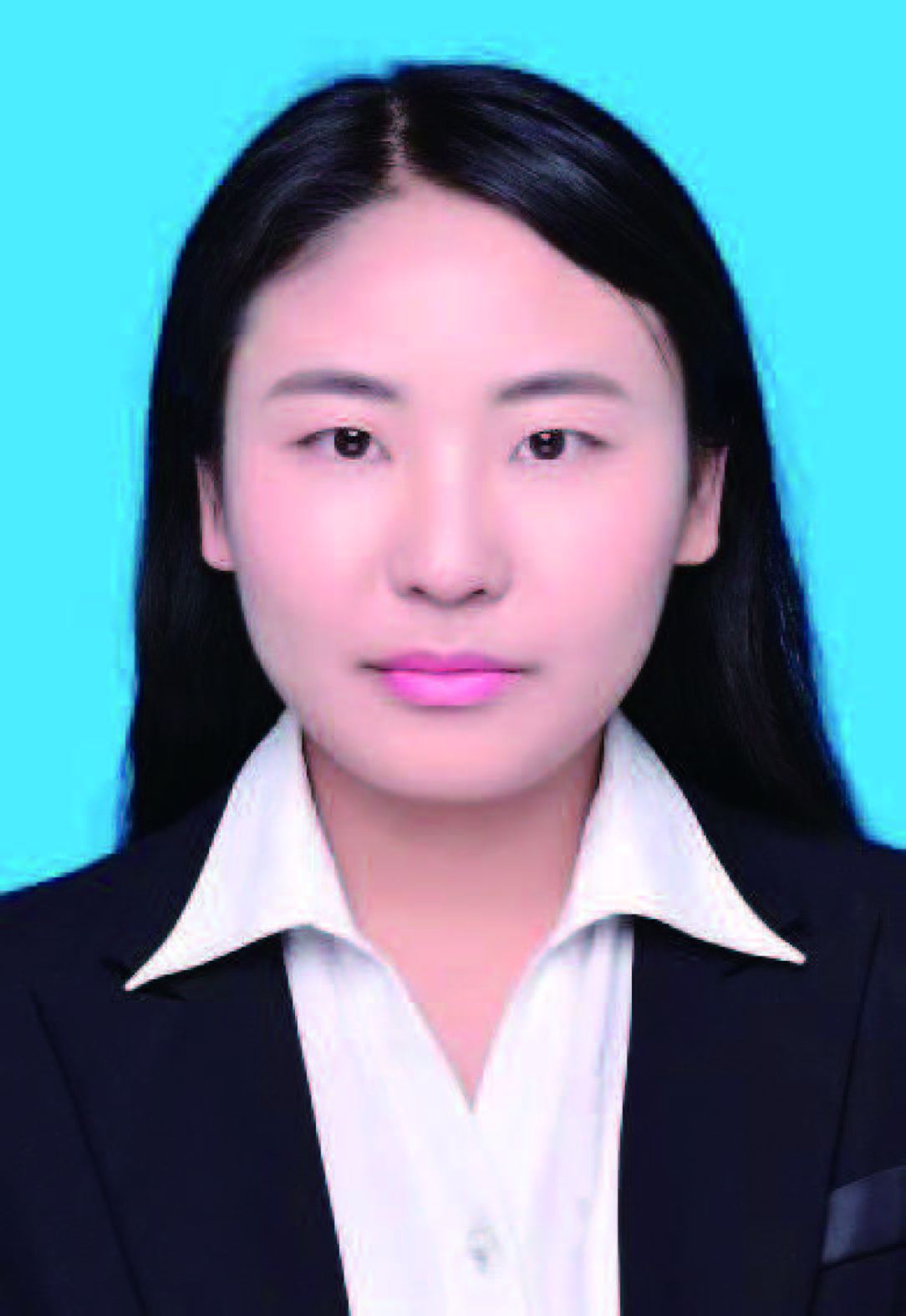}}]{Zhuoran Ma}
received the B.E. degree from the School of Software Engineering, Xidian University, Xi’an, China, in 2017. She is currently a Ph.D candidate with the Department of Cyber Engineering, Xidian University. Her current research interests include data security and secure computation outsourcing.
\end{IEEEbiography}

\begin{IEEEbiography}[{\includegraphics[width=0.9in,height=1.25in,clip,keepaspectratio]{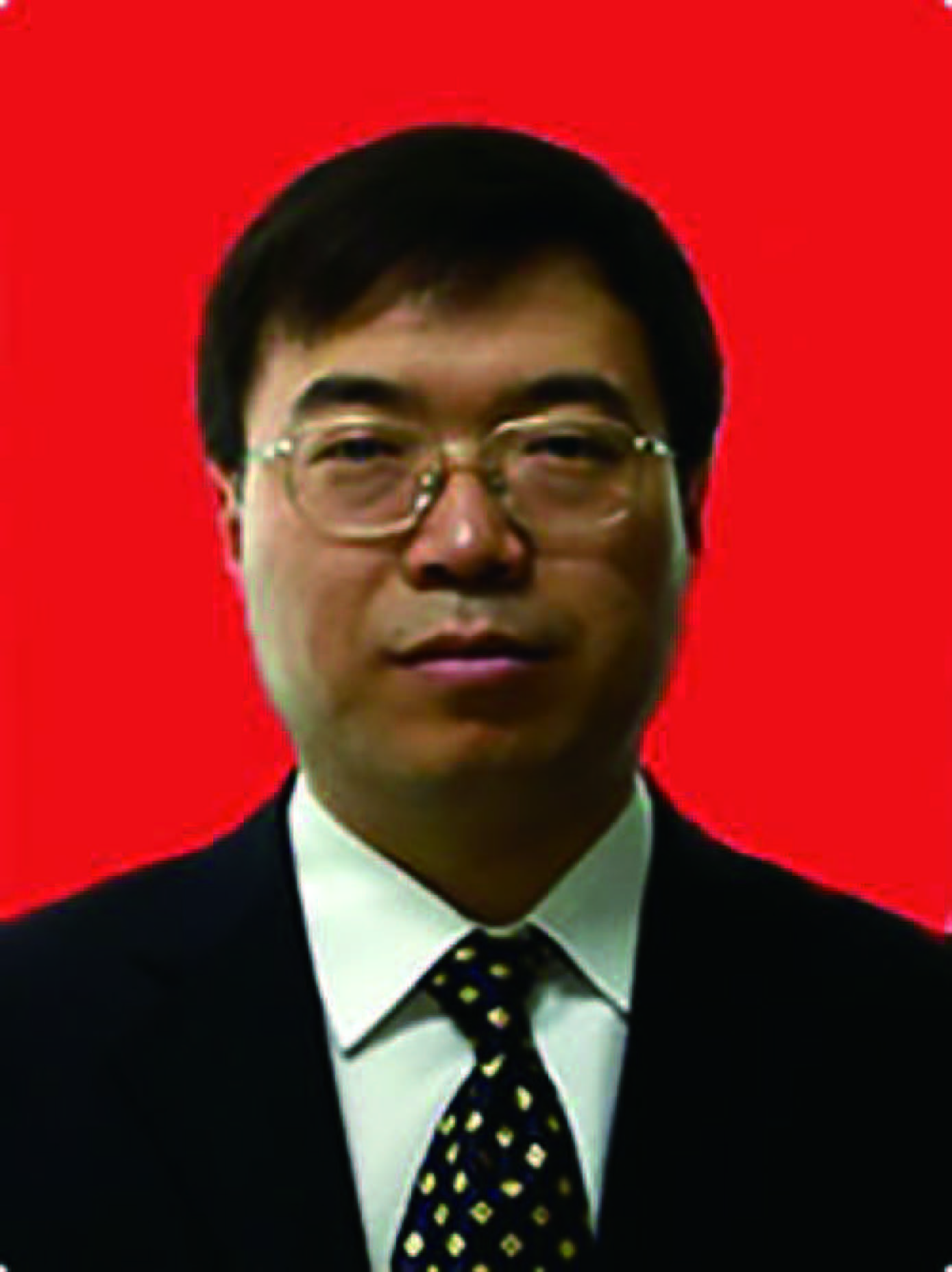}}]{Jianfeng Ma}
 received the B.S. degree in mathematics from Shaanxi Normal University, Xi'an, China, in 1985, and the M.S. degree and the Ph.D. degree in computer software and telecommunication engineering from Xidian University, Xi'an, China, in 1988 and 1995, respectively. From 1999 to 2001, he was a Research Fellow with Nanyang Technological University of Singapore. He is currently a professor and a Ph.D. Supervisor with the Department of Computer Science and Technology, Xidian University, Xi'an, China. He is also the Director of the Shaanxi Key Laboratory of Network and System Security. His current research interests include information and network security, wireless and mobile computing systems, and computer networks.
\end{IEEEbiography}

\begin{IEEEbiography}[{\includegraphics[width=0.9in,height=1.25in,clip,keepaspectratio]{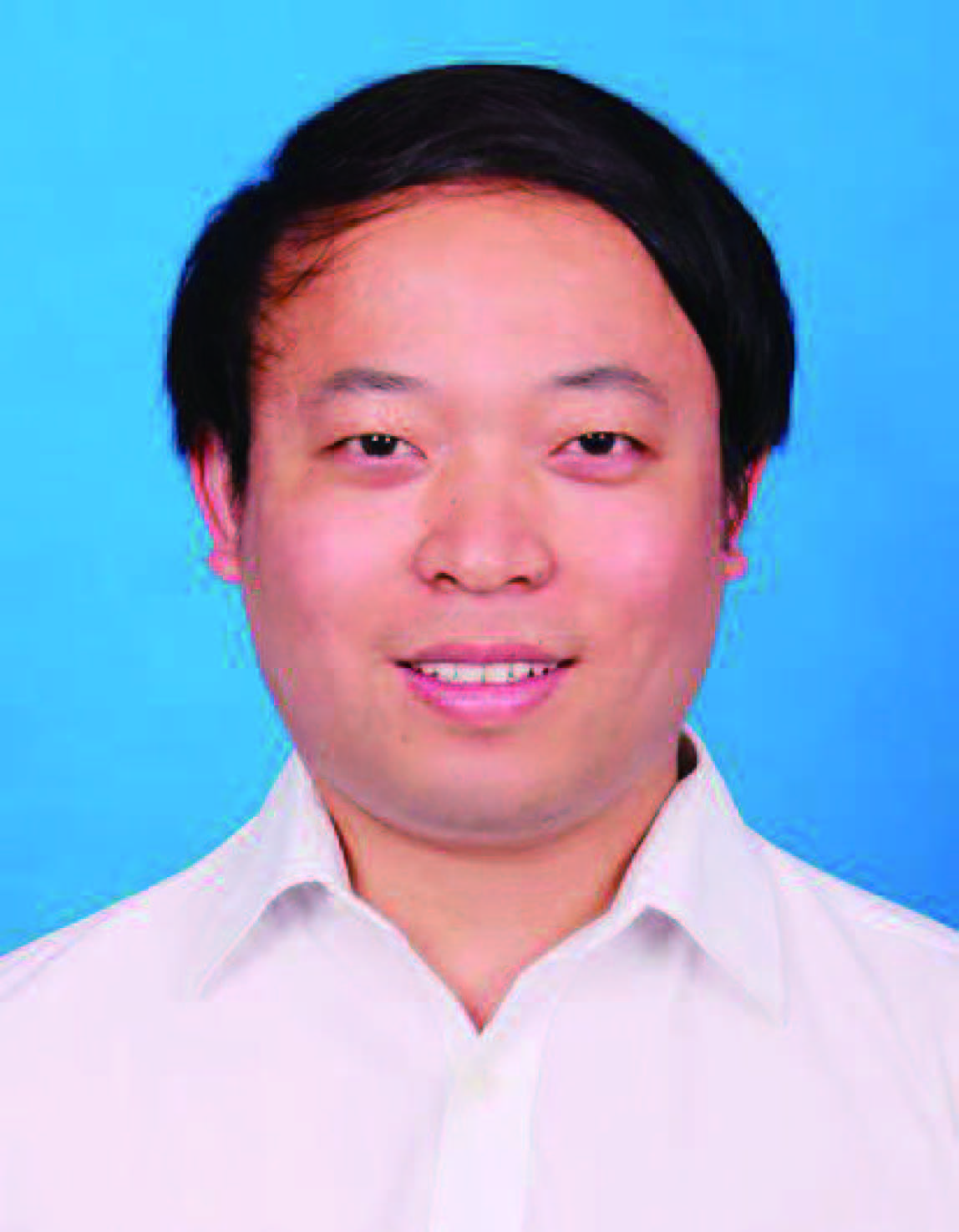}}]{Yinbin Miao}
 received the B.E. degree with the Department of Telecommunication Engineering from Jilin University, Changchun, China, in 2011, and Ph.D. degree with the Department of Telecommunication Engineering from Xidian university, Xi'an, China, in 2016. He is currently a Lecturer with the Department of Cyber Engineering in Xidian university, Xi'an, China. His research interests include information security and applied cryptography.
\end{IEEEbiography}

\begin{IEEEbiography}[{\includegraphics[width=0.9in,height=1.25in,clip,keepaspectratio]{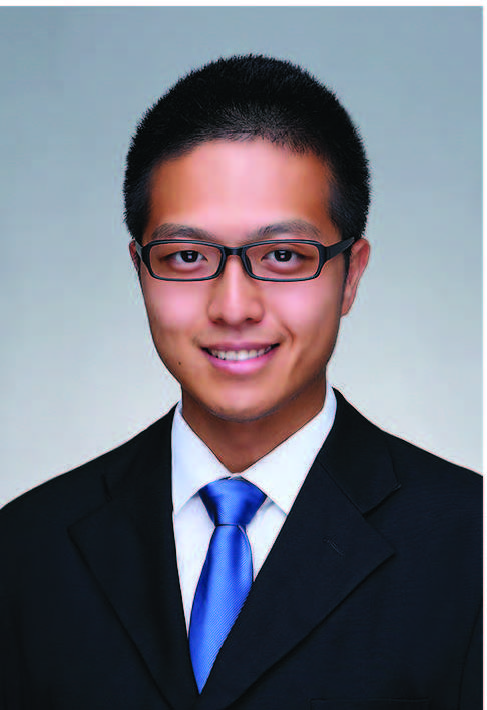}}]{Ximeng Liu} (S’13-M’16) received the B.Sc. degree in electronic engineering from Xidian University, Xi’an, China, in 2010 and the Ph.D. degree in Cryptography from Xidian University, China, in 2015. Now he is the full professor in the College of Mathematics and Computer Science, Fuzhou University. Also, he was a research fellow at the School of Information System, Singapore Management University, Singapore. He has published more than 200 papers on the topics of cloud security and big data security including papers in IEEE TOC, IEEE TII, IEEE TDSC, IEEE TSC, IEEE IoT Journal, and so on. He awards “Minjiang Scholars” Distinguished Professor, “Qishan Scholars” in Fuzhou University, and ACM SIGSAC China Rising Star Award (2018). His research interests include cloud security, applied cryptography and big data security. He is a member of the IEEE, ACM, CCF.
\end{IEEEbiography}

\begin{IEEEbiography}[{\includegraphics[width=1in,height=1.25in,clip,keepaspectratio]{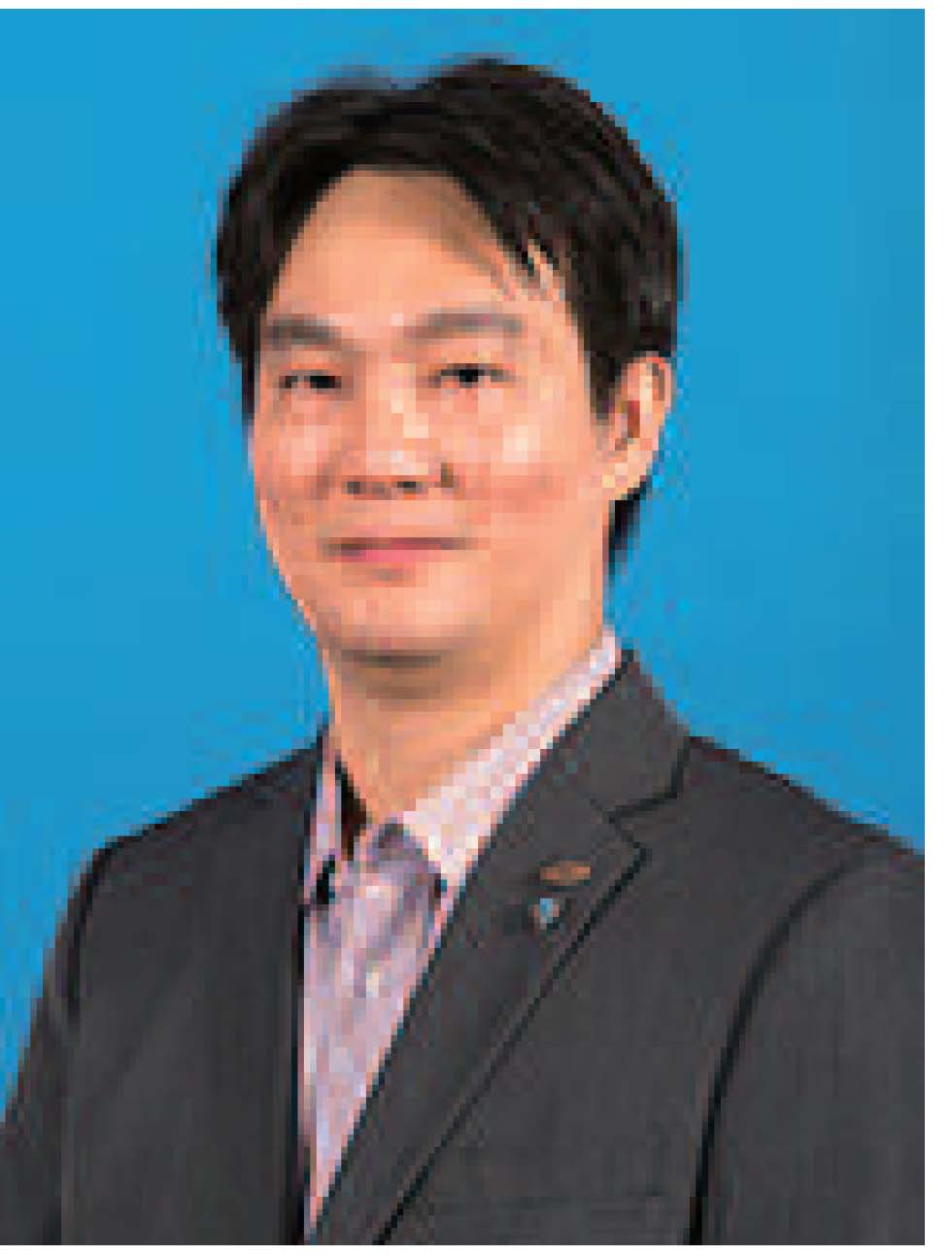}}]{Kim-Kwang Raymond Choo}
(SM'15) received the Ph.D. in Information Security in 2006 from Queensland University of Technology, Australia. He currently holds the Cloud Technology Endowed Professorship at The University of Texas at San Antonio (UTSA). He is the recipient of various awards including the UTSA College of Business Col. Jean Piccione and Lt. Col. Philip Piccione Endowed Research Award for Tenured Faculty in 2018, ESORICS 2015 Best Paper Award. He is an Australian Computer Society Fellow, and an IEEE Senior Member.
\end{IEEEbiography}

\begin{IEEEbiography}[{\includegraphics[width=1in,height=1.25in,clip,keepaspectratio]{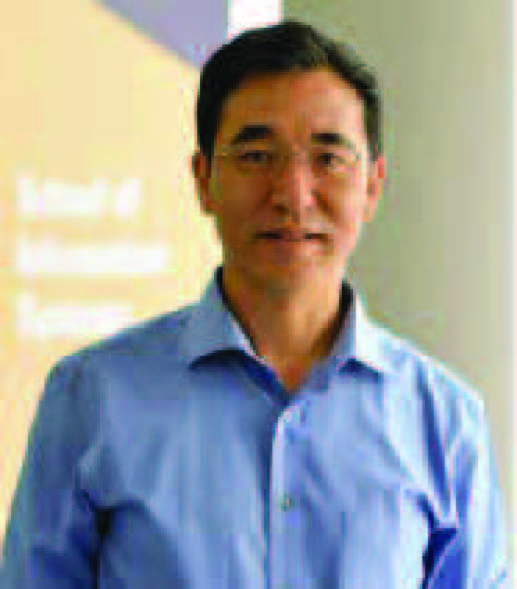}}]{Robert H. Deng}
(F'16) has been a Professor with the School of Information Systems, Singapore Management University, since 2004. His research interests include data security and privacy, multimedia  security, and network and system security. He was  an Associate Editor of the IEEE TRANSACTIONS ON {INFORMATION FORENSICS AND SECURITY} from  2009 to 2012. He is currently an Associate Editor of the IEEE TRANSACTIONS ON DEPENDABLE AND  SECURE COMPUTING and Security and Communication Networks (John Wiley). He is the cochair of the Steering Committee of the ACM Symposium on Information, Computer and Communications Security. He received the University Outstanding Researcher Award from the National University of Singapore in 1999 and the Lee Kuan Yew Fellow for Research Excellence from Singapore  Management University in 2006. He was named Community Service Star and Showcased Senior Information Security Professional by (ISC)$^2$ under its Asia-Pacific Information Security Leadership Achievements Program in 2010.
\end{IEEEbiography}

\end{document}